\documentclass[11pt]{article} 
\pdfoutput=1

\def\focs{0}
\def\comments{0}
\newcommand{\full}[1]{\ifnum\focs=0 {#1} \fi}
\newcommand{\focsfull}[2]{\ifnum\focs=1{#1}\else{#2}\fi}

\usepackage{amsmath,amsfonts,amssymb,latexsym,amsthm,paralist,wrapfig,environ}
\usepackage{color}
\usepackage[usenames,dvipsnames,svgnames,table]{xcolor}

\usepackage[footnotesize,it]{caption}
\usepackage{graphicx} 
\usepackage{wrapfig}

\usepackage{fullpage}

\usepackage[pagebackref,colorlinks=true,urlcolor=blue,linkcolor=blue,citecolor=black,filecolor=black]{hyperref}
\usepackage{algorithm,algorithmic}

\newcommand{\eps}{\varepsilon}
\renewcommand{\epsilon}{\varepsilon}

\newcommand{\bit}[1]{\{0,1\}^{{#1}}}
\newcommand{\closedef}{{\hfill $\Box$}}
\newcommand{\from}{\leftarrow}

\newtheorem{theorem}{Theorem}[section]
\newtheorem{lemma}[theorem]{Lemma}
\newtheorem{prop}[theorem]{Proposition}
\newtheorem{cor}[theorem]{Corollary}
\newtheorem{observation}[theorem]{Observation}

\newtheorem{claim}[theorem]{Claim}
\newtheorem{defn}[theorem]{Definition}

\newtheorem{remark}{Remark}

\newcommand{\comment}[1]{{\makebox[0.5in][r]{$\rhd$} \parbox[t]{0.63\textwidth}{\footnotesize {#1}}}}
\newcommand{\scomment}[1]{\makebox[0.5in][r]{$\rhd$} {\footnotesize {#1}}}
\newcommand{\wcomment}[1]{{\makebox[0.5in][r]{$\rhd$} \parbox[t]{0.83\textwidth}{\footnotesize {#1}}}}

\ifnum\focs=1
\newcommand{\boldcomment}[1]{}
\else
\newcommand{\boldcomment}[1]{\bigskip \hrule\smallskip{\bf {
      {#1}}}\vspace*{1ex}}
\fi

\renewcommand{\le}{\leqslant}
\renewcommand{\leq}{\leqslant}
\renewcommand{\ge}{\geqslant}
\renewcommand{\geq}{\geqslant}

\newcommand{\bsc}{\mathsf{BSC}}
\newcommand{\adv}{\mathsf{ADV}}

\newcommand{\F}{{\mathbb F}}
\newcommand{\wt}{{\rm wt}}
\newcommand{\RS}{\mathsf{RS}}
\newcommand{\knr}{{\sf KNR}}
\newcommand{\twise}{\mathsf{POLY}}
\newcommand{\SC}{{\sf SC}}
\newcommand{\Samp}{\mathsf{Samp}}
\newcommand{\BSC}{\mathsf{BSC}}
\newcommand{\Dec}{\mathsf{Dec}}
\newcommand{\REC}{{\mathsf{REC}}}
\newcommand{\nis}{\mathsf{Nis}}
\newcommand{\GEN}{\mathsf{Gen}}

\newcommand{\PRC}{\mathsf{PRC}}

\newcommand{\BPG}{\ensuremath\text{\sf BPG}^{\mathsf{Nisan}}}
\newcommand{\LSC}{\ensuremath\text{\sf LSC}}
\newcommand{\LG}{\lsg}
\newcommand\lsg{\mathsf{BPPRG}} 

\newcommand{\polyprg}{\mathsf{PolyPRG}}

\newcommand{\WC}{\ensuremath\text{\sf WC}}
\newcommand{\xor}{\oplus}
\newcommand{\A}{{\ensuremath{\mathcal{A}}}}
\newcommand{\B}{{\ensuremath{\mathcal{B}}}}

\newcommand{\close}[2]{\ {\stackrel{{#2}}{\approx}}_{#1} \ }

\newcommand{\ctrlset}{V}
\newcommand{\T}{{\mathsf{T}}}
\newcommand{\spcbound}{{S}}

\newcommand{\channel}{\mathsf{W}}
\newcommand{\E}{\mathcal{E}}
\newcommand{\ctrl}{\mathsf{ctrl}}

\newcommand{\dat}{\mathsf{data}}

\newcommand{\hamdist}{\mathrm{dist}}

\newcommand{\poly}{\operatorname{poly}}

\newcommand{\Enc}{\ensuremath{\mathsf{Enc}}}

\newcommand{\expec}{{\mathbb E}}

\newcommand{\defeq}{\overset{\text{def}}{=}}

\newcommand{\mynote}[2]{\marginpar{\tiny {\bf {#2}:} \sf {#1}}}

\ifnum\comments=1
\newcommand{\vnote}[1]{\mynote{#1}{VG}}
\newcommand{\anote}[1]{\mynote{#1}{ADS}}
\NewEnviron{notes}
{
  \color{blue}
  \BODY
}
{
}
\else
\newcommand{\vnote}[1]{}
\newcommand{\anote}[1]{}
\NewEnviron{notes}%This environment swallows its argument
{
}
{
}

\fi

\newcommand{\mypar}[1]{\medskip\noindent\textbf{#1.}}
\newcommand{\myabspar}[1]{\emph{#1:}}

\title{Optimal Rate Code Constructions for Computationally Simple
  Channels
\thanks{An extended abstract appeared in the Proceedings of the
  \emph{51st Annual IEEE Symposium on Foundations of Computer Science}
  (FOCS), October 2010.}}
\author{Venkatesan Guruswami\thanks{Computer Science Department, Carnegie Mellon University. 
V. G. was supported by a Packard Fellowship
 and NSF grants CCF-0953155 and CCF-0963975. Email: {\tt
   guruswami@cmu.edu}.} \and Adam Smith\thanks{Computer Science \&
 Engineering Department, Pennsylvania State University. A. S. was supported by NSF grants CCF-0747294 and CCF-0729171. Email: {\tt asmith@psu.edu}.}}

\begin{document}

\maketitle

\begin{abstract}
  We consider coding schemes for \emph{computationally bounded}
  channels, which can introduce an arbitrary set of errors as long as
  (a) the fraction of errors is bounded with high probability by a
  parameter $p$ and (b) the process which adds the errors can be
  described by a sufficiently ``simple'' circuit.  Codes for such
  channel models are attractive since, like codes for standard
  adversarial errors, they can handle channels whose true behavior is
  \emph{unknown} or \emph{varying} over time.

  For two classes of channels, we provide explicit, efficiently
  encodable/decodable codes of optimal rate where only
  \emph{in}efficiently decodable codes were previously known. In each
  case, we provide one encoder/decoder that works for \emph{every}
  channel in the class.  \ifnum\focs=0 The encoders are randomized,
  and probabilities are taken over the (local, unknown to the decoder)
  coins of the encoder and those of the channel.  \fi

  \myabspar{Unique decoding for additive errors} We give the first
  construction of a polynomial-time encodable/decodable code for
  \emph{additive} (a.k.a. \emph{oblivious}) channels that achieve the
  Shannon capacity $1-H(p)$.  \ifnum\focs=0 These are channels which
  add an arbitrary error vector $e\in\bit{N}$ of weight at most $pN$
  to the transmitted word; the vector $e$ can depend on the code but
  not on the particular transmitted word.  \fi \full{Such channels
    capture binary symmetric errors and burst errors as special
    cases.}

% \myabspar{List-decoding for online log-space channels} 
% \ifnum\focs=0%
% A \emph{space-$\spcbound(N)$
%   bounded} channel reads and modifies the transmitted codeword as a
% stream, using at most $\spcbound(N)$ bits of workspace on transmissions of $N$
% bits. For constant $\spcbound$, this captures many models from the literature,
% including \emph{discrete channels with finite memory} and
% \emph{arbitrarily varying channels}. We give an efficient code with
% optimal rate (arbitrarily close to $1-H(p)$) that recovers a short list containing
% the correct message with high probability for channels limited to
% \emph{logarithmic} space.
% \else%
% We give an efficient code with
% optimal rate (arbitrarily close to $1-H(p)$) that recovers a short list containing
% the correct message with high probability for  channels which
% read and modify  the transmitted codeword as a
% stream, using at most $O(\log N)$ bits of workspace on transmissions of $N$
% bits. \full{This class captures many models from the literature,
% including \emph{discrete channels with finite memory} and
% \emph{arbitrarily varying channels}.}
% \fi%

\myabspar{List-decoding for polynomial-time channels} For every
constant $c>0$, we give a Monte Carlo construction of an code with
optimal rate (arbitrarily close to $1-H(p)$) that efficiently recovers a short list containing
the correct message with high probability for channels describable by
circuits of size at most $N^c$. We are not aware of any channel models considered in the
information theory literature, other than purely adversarial channels,
which require more than linear-size circuits to implement. We justify
the relaxation to list-decoding with an impossibility result showing
that, in a large range of parameters ($p>1/4$), codes that are uniquely decodable
for a modest class of channels (online, memoryless, nonuniform
channels) cannot have positive rate. 

\end{abstract}

% \category{E.4}{Coding and Information Theory}{Error control codes}
% \category{F.1.3}{Computation by Abstract Devices}{Complexity Measures and Classes}
% \terms{Algorithms; Theory} 
% \keywords{Computationally bounded channels, error-correcting codes, 
%   pseudorandomness.} 
% \acmformat{}

% \newpage
% \tableofcontents
% \newpage

\section{Introduction}
\label{sec:intro}

The theory of error-correcting codes has had two divergent schools of thought, dating back to its origins, based on the underlying model of the noisy channel.  Shannon's theory \cite{shannon} modeled the channel as a stochastic process with a known probability law. Hamming's work \cite{hamming} suggested a worst-case/adversarial model, where the channel is subject only to a limit on the number of errors it may cause. 

These two approaches share several common tools, however in terms of quantitative results, the classical results in the harsher Hamming model are much weaker. 
%For example, when transmitting bits, error recovery from more than a fraction 1/4 of worst-case errors is not possible, whereas even close to a fraction 1/2 of random errors can be corrected.
%
For instance, for the binary symmetric channel %$\bsc_p$ 
which flips each transmitted
bit independently with probability $p < 1/2$, the optimal rate of
reliable transmission is known to be the Shannon capacity $1-H(p)$,
where $H(\cdot)$ is the binary entropy function \cite{shannon}. Concatenated codes (Forney~\cite{forney}) and polar codes
(Arikan~\cite{Arikan09}) can transmit at rates arbitrarily close to this capacity and are
efficiently 
decodable.
In contrast, for
adversarial channels that can corrupt up to a fraction $p$ of symbols
in an {\em arbitrary} manner, the optimal rate is unknown in general,
though it is known for all $p \in (0,\frac12)$ that the rate has to be much smaller than the Shannon
capacity. In fact, for $p \in [\frac14,\frac12)$, the achievable rate over an
adversarial channel is asymptotically zero, while the Shannon capacity
$1-H(p)$ remains positive.

% But what if the errors
% are {\em adversarial} and not randomly distributed? For the
% adversarial channel $\adv_p$ where the channel can corrupt up to a
% fraction $p$ of symbols in an {\em arbitrary} manner, it is known that
% for error-free communication to be possible, the rate has to be much
% smaller than the Shannon capacity $1-H(p)$. For $p > 1/4$, this rate
% goes to zero, and for $p < 1/4$ the best rate known to date, even by
% non-constructive methods, equals $1-H(2p)$ (via the Gilbert-Varshamov
% bound). 
%
%Determining the best asymptotic rate for error fraction $p$
%(equivalently, minimum relative distance $2p$) remains an important
%open question in combinatorial coding theory.

% \begin{wrapfigure}[11]{r}{2in}
%   \vspace{-4pt}
%   \includegraphics[width=2in]{coding-bounds-captions}  
%   \caption{\textit{Bounds on the achievable rate for adversarial and random errors.}}
% \label{bounds}
% \end{wrapfigure}

% There are two common motivations for studying codes that tolerate for
% adversarial errors. First, their combinatorial structure makes them
% extremely useful building blocks in other areas of combinatorics
% (\emph{e.g.}, authentication schemes, quorum systems, designs,
% pseudorandom generators, block ciphers, etc). Second,  

Codes that tolerate worst-case
errors are attractive because they assume nothing about the distribution of the errors introduced by the channel, only a bound on the number of errors. Thus, they 
can be used to transmit information reliably over a large range of channels whose true
behavior is unknown or varies over time. In contrast, codes tailored to a specific channel model tend to fail when the model changes. For example,
concatenated codes with a high rate outer code, which can transmit efficiently and reliably at the
Shannon capacity with i.i.d.~errors, fail miserably in the presence of
\emph{burst} errors that occur in long runs.

In this paper, we consider several intermediate models of uncertain 
channels, as a meaningful and well-motivated bridge between the Shannon and Hamming models.
Specifically, we consider \emph{computationally bounded} channels,
which can introduce an arbitrary set of errors as long as (a) the total
fraction of errors is bounded by $p$ with high probability and
(b) the process which adds the errors can be described by a
sufficiently ``simple'' circuit. The idea and motivation behind these models is that
natural processes may be mercurial, but are not computationally intensive. These models are powerful enough to capture natural settings
like \emph{i.i.d.} and burst errors, but weak enough to allow
efficient communication arbitrarily close to the Shannon
capacity $1-H(p)$. The models we study, or close variants, have been considered
previously---see Section~\ref{sec:related} for a discussion of related
work. The computational perspective we espouse is inspired by the
works of Lipton~\cite{lipton} and Micali \emph{et al.}~\cite{MPSW}.

For two classes of channels, we provide efficiently
encodable and decodable codes of optimal rate (arbitrarily close to $1-H(p)$) where only
\emph{in}efficiently decodable codes were previously
known. In each case, we provide one encoder/decoder that works for
\emph{every} channel in the class. In particular, our results apply
even when the channel's  behavior depends on the code.

% We have results for three classes of channels. For each, we provide
% the first polynomial-time encoding and decoding algorithms that use
% $N$ binary symbols to encode a message of length $N(1-H(p)-\eps)$, for
% arbitrarily small $\eps>0$, thereby matching the capacity of the
% $\bsc_p$ channel. Each encoder/decoder performs reliably for
% \emph{every} channel in the class. We stress that our results make
% {\em no} setup assumptions between the encoder and decoder.

\mypar{Structure of this paper}
We first describe the models and our results briefly
(Section~\ref{sec:briefresults}), and outline our main technical
contributions (Section~\ref{sec:techniques}). In Section
\ref{sec:related}, we describe %highlight the relationship of our work to
related lines of work %which were 
aimed at handling (partly)
adversarial errors with rates near Shannon capacity.  Our results are
stated formally in Section \ref{sec:intro-results}.
Section~\ref{sec:ld-implies-avc} describes our list-decoding based code
constructions for recovery from additive
errors. As the needed list-decodable codes are not explicitly known, this only gives an existential result. Section~\ref{sec:overview} describes the approach behind our efficient
constructions at a high-level. The remainder of the paper describes
and analyzes the constructions in detail, in order of increasing
channel strength: additive errors
(Section~\ref{sec:efficient-oc-codes}), % space-bounded errors
% (Section~\ref{sec:logspace}) 
and time-bounded errors
(Section~\ref{sec:time-bounded}). 
The appendices contain extra details
on the building blocks in our constructions (\ref{sec:ingredients}), results for the ``average'' error criterion
(\ref{app:codes-avg-error}) and our impossibility result showing the necessity of list decoding for error fractions exceeding $1/4$
(\ref{sec:imposs}), respectively (see Section \ref{sec:intro-results} for formal statements of these claims).

\subsection{Our results}\label{sec:briefresults}

The encoders we construct are \emph{stochastic} (that is, randomized). Probabilities
are taken over the (local, unknown to the decoder) coins of the
encoder and the choices of the channel; messages may be chosen
adversarially and known to the channel.  Our results
do not assume any setup or shared randomness between the encoder and decoder.

\mypar{Unique decoding for additive channels} 
We give the first explicit construction of stochastic codes with
polynomial-time encoding/decoding algorithms that approach the
Shannon capacity $1-H(p)$ for \emph{additive}
(a.k.a. \emph{oblivious}) channels. These are channels which add an
arbitrary error vector $e\in\bit{N}$ of Hamming weight at most $pN$ to the
transmitted codeword (of length $N$). The error
vector may depend on the code and the message but, crucially, not on the encoder's
local random coins.
% For any desired $p \in (0,1/2)$ and $\eps > 0$,
% we construct a polynomial time stochastic encoder $\Enc$ mapping
% messages of $(1-H(p)-\eps)$ bits into $n$ bits and a polynomial time
% decoder $\Dec$ such that for every message $m \in \bit{Rn}$ and every
% error vector $e \in \bit{N}$ of Hamming weight at most $pN$, with high
% probability over the randomness $r$ used at the encoder, the decoding
% succeeds, i.e., $\Dec(\Enc(m,r)+e)=m$.  
% Note that the decoder is {\em
%   not} given the randomness $r$ used at the encoder; its only input is
% the corrupted codeword.
Additive errors capture binary symmetric errors as well
as certain models of correlated errors, like burst errors.
For a deterministic encoder, the
additive error model is equivalent to the usual adversarial error
model. A randomized encoder is thus necessary to achieve the Shannon capacity.

We also provide a novel, simple proof that (inefficient)
capacity-achieving codes \emph{exist} for additive channels. We do so
by combining linear list-decodable codes with rate approaching $1-H(p)$ (known to exist, but not known
to be efficiently decodable) with a
special type of authentication scheme. Previous existential proofs
relied on complex random coding arguments~\cite{CN88a,langberg08}; see
the discussion of related work below.

% Csisz\'ar and Narayan \cite{CN88a} and Langberg~\cite{langberg08}
% proved the existence of codes for such channels via random coding
% arguments. In addition to our polynomial time construction, we provide
% a simpler proof that (inefficient) capacity-reaching codes for this
% channel model exist. Our construction combines linear list-decodable
% codes with a special type of authentication scheme. Such codes are
% then used as building blocks for our more complicated polynomial time
% construction.

\mypar{Necessity of list decoding for richer channel models} The additive errors
model postulates that the error vector has to be picked obliviously,
before seeing the codeword.  To model more complex processes, we will allow the channel to see the codeword and decide upon the error as a function of the codeword. We will stipulate a computational restriction on how the channel might decide what error pattern to cause. The simplest restriction (beyond the additive/oblivious case) is a ``memoryless" channel that decides the action on the $i$'th bit based only on that bit (and perhaps some internal state that is oblivious to the codeword). 

%
%Consider a channel that processes the codeword as a stream, deciding
%as it goes which positions to corrupt. The channel's only limitation
%is a bound $\spcbound(N)$ on the amount of work space it can use (as a
%function of the block length $N$).
%%
%Roughly, we view the channel as a
%finite automaton with $2^{\spcbound(N)}$ states.
%More precisely, in order to allow
%\emph{nonuniform} dependency on the code, we model the channel
%as a width-$2^{\spcbound(N)}$ branching program that outputs one bit for every
%input bit that it reads.
%%
%Even for constant space $\spcbound$, this model captures a wide range of
%channels considered in coding theory, including  additive channels, discrete channels with
%  finite memory, bounded delay and \emph{arbitrarily varying
%    channels}; see the discussion of
%related work below for definitions. 
%Our constructions tolerate the larger class of logarithmic-space
%channels. As above, we assume that the channel introduces at most
%$pN$ errors with high probability.

  First, we show that even against such memoryless channels,
  reliable \emph{unique} decoding with positive
  rate is impossible when $p>1/4$. The idea is that even a
  memoryless adversary can make the transmitted codeword difficult to
  distinguish from a different, random codeword. The proof
  relies on the requirement that a single code must work for all
  channels, since the ``hard'' channel depends on the code.

  Thus, to communicate at a rate close to $1-H(p)$ for all $p$, we
  consider the relaxation to \emph{list-decoding}: the decoder is
  allowed to output a small list of messages, one of which is
  correct. List-decodable codes with rate approaching $1-H(p)$ are
  known to exist even for adversarial errors
  \cite{ZP,elias91}. However, constructing efficient (\emph{i.e.},
  polynomial-time encodable and decodable) binary codes for list
  decoding with near-optimal rate is a major open
  problem.\footnote{Over large alphabets, explicit optimal rate
    list-decodable codes {\em are} known.}  Therefore, our goal is to
  exploit the computational restrictions on the channel to get
  efficiently list-decodable codes.

\mypar{Computationally bounded channels}
  We consider channels whose
  behavior on $N$-bit inputs is described by a circuit of size
  $\T(N)$
  (for example $\T(N) = O(N^2)$). We sometimes call this the
  \emph{time-bounded} model, and refer to $\T(\cdot)$ as a time bound. We do not know of any channel models
  considered in the information theory literature, other than purely
  adversarial channels, which require more than linear time to
  implement. 

We also discuss \emph{(online) space-bounded} channels; these
  channels make a single pass over the transmitted codeword, deciding
  which locations to corrupt as they go, and are
  limited to storing at most $\spcbound(N)$
  bits (that is, they can be describe by one-pass branching programs
  of width at most $2^{\spcbound(N)}$). Logarithmic space channels, in
   particular, can be realized by polynomial-size circuits.

% Returning to the issue of channel models, an extension of the memoryless channel is to consider a space-bounded channel, for instance an (online) logspace-bounded channel (which has $O(\log N)$ bits of memory when transmitting $N$ bit codewords). Specifically, the channel will process the codeword as a stream, remembering details of the codeword to the extent its memory limitation allows, and using its memory state to guide its decision on whether or not to flip a particular bit. 
% More generally,
%   one may consider channels whose behavior on $N$-bit inputs is
%   described by a circuit of size $\T(N)$ (for example $\T(N) = O(N^2)$). Logarithmic space channels, in
%   particular, can be realized by polynomial-size circuits. In fact, we
%   do not know of any channel models considered in the information
%   theory literature, other than purely adversarial channels, which
%   require more than linear time to
%   implement.

  \mypar{List decoding for polynomial time channels} 
  Our main contribution for time-bounded channels is a construction of polynomial-time encodable and
  list-decodable codes
  % \footnote{In standard list-decoding, the goal is to find all words
  %   within distance $pN$ of any given string. Our algorithms recover
  %   a small list of messages containing the real message, assuming
  %   that the received word was produced by running a real codeword
  %   through a channel from a specific class.}
  that approach the optimal rate for channels whose time bound is a polynomial in 
    the block length $N$. Specifically, we give an efficiently computable encoding function, $\Enc$, that stochastically encodes the message $m$ into $\Enc(m;r)$ where $r$ is {\em private} randomness chosen at the encoder, such that 
    for every message $m$ and time $N^c$-bounded channel $\channel$, the decoder takes $\channel(\Enc(m;r))$
  as input and returns a small list of messages that, with high
  probability over $r$ and the coins of the channel, contains the real
  message $m$. The size of the list is polynomial in $1/\eps$, where
  $N(1-H(p)-\eps)$ is the length of the transmitted messages. We stress that the decoder does {\em not} know the choice of random bits $r$ made at the encoder.

  The construction of our encoding function $\Enc(\cdot,\cdot)$ is
  {\em Monte Carlo} --- we give a randomized construction of $\Enc$
  that (together with the decoder) offers the claimed properties (for
  all messages and $N^c$-bounded channels) with high probability. The
  sampling step is polynomial time, and the resulting function
  $\Enc(m;r)$ can be computed from $m,r$ in deterministic polynomial
  time. However, we do not know how to efficiently check and certify
  that a particular random choice is guaranteed to have the claimed
  properties.  Obtaining fully explicit constructions remains an
  intriguing open
  problem. % In the absence of explicit constructions, this is a useful compromise that offers the functional benefits of an efficient coding scheme.
  In the case of polynomial time bounded channels, it seems
  unlikely that solve the problem without additional complexity
  assumptions (such as the existence of certain pseudorandom
  generators). In the case of  online logspace-bounded channel, it may
  be possible to obtain fully explicit constructions without
  complexity assumptions (using, for example, Nisan's pseudorandom generators for logspace). We will elaborate on this aspect
  in Section \ref{sec:openq}.

One technicality is that the decoder need {\em not} return all words within a
  distance $pN$ of the received word (as is the case for the standard
  ``combinatorial'' notion of list decoding),
  but it {\em must} return the correct message as one of the candidates with
  high probability.  This notion of
  list-decoding is natural for stochastic codes in communication applications. In fact, it is
  exactly the notion that is needed in constructions which ``sieve''
  the list, such as \cite{Gur-sidechannel,MPSW}, and one application of our results is a
  \emph{uniquely} decodable code for the public-key model (assuming a
  specific polynomial time bound), strengthening results of
  \cite{MPSW}. See the discussion of related work in
  Section~\ref{sec:related} for a more precise statement.

% Our results raise a compelling question: are there stochastic
%    codes of rate approaching $1-H(p)$ that can be uniquely decoded
%    from $pn$ log-space errors, when $p<1/4$?

%
%  Our construction of list-decodable codes for logarithmic-space
%  channels can
%  be extended to handle channels with a given polynomial time bound
%  $\T(n)=N^c$, for any fixed $c>1$, under an additional assumption, namely, the existence of pseudorandom generators of
%  constant stretch that output $N$ pseudorandom bits and fool
%  circuits of size $N^c$. Such generators exist, for example, if there
%  are functions in $\mathsf{E}$ which have no subexponential-size
%  circuits~\cite{NW94,IW97}, or if one-way functions exist~\cite{yao82,HILL}.

\medskip For both the additive and time-bounded models, our constructions require new methods for applying tools from cryptography and
  derandomization to coding-theoretic problems. We give
  a brief high-level discussion of these techniques next. An expanded overview 
  the approach behind our code construction appears in Section~\ref{sec:overview}.

\subsection{Techniques}
\label{sec:techniques}

\mypar{Control/payload construction} In our constructions, we develop
several new techniques. The first is
a novel ``reduction'' from the standard coding setting with no setup
to the setting of \emph{shared secret randomness}. In
models in which errors are distributed evenly, such a reduction is
relatively simple~\cite{ahlswede}; however, this reduction fails
against adversarial errors. Instead, we show how to \emph{hide} the
secret randomness (the \emph{control information}) inside the main
codeword (the \emph{payload}) in such a way that the decoder can learn
the control information but (a) the control information remains hidden to a
bounded channel and (b) its encoding is robust to a certain, weaker class of
errors. We feel this technique should be useful in other settings of
bounded adversarial behavior.

Our reduction can also be viewed as a novel way of bootstrapping
from ``small'' codes, which can be decoded by brute force, to
``large'' codes, which can be decoded efficiently. The standard way to
do this is via concatenation; unfortunately, concatenation does not
work well even against mildly unpredictable models, such as the additive
error model.

\mypar{Pseudorandomness} Second, our results further develop a
surprising connection between coding and
pseudorandomness. \emph{Hiding} the ``control information'' from the
channel requires us to make different settings of the control information
\emph{indistinguishable} from the channel's point of view. Thus, our
proofs apply techniques from cryptography together with constructions
of pseudorandom objects (generators, permutations, and samplers) from
derandomization. Typically, the ``tests'' that must be fooled are
compositions of the channel (which we assume has low complexity) with
some subroutine of the decoder (which we \emph{design} to have low
complexity). The connection to pseudorandomness appeared in a simpler
form in the previous work on bounded channels~\cite{lipton,GLYY,MPSW};
our use of this connection is significantly more delicate.

The necessary pseudorandom permutations and samplers can be explicitly constructed, and the construction of the needed complexity-theoretic pseudorandom generators is where we uses randomness in our construction. For the case of additive errors, information-theoretic objects that we can construct deterministically efficiently (such as $t$-wise independent strings) suffice and we get a fully explicit construction. 

\section{Background and Related Previous Work}
\label{sec:related}
There are several lines of work aimed at handling adversarial, or partly adversarial, errors with rates near the
Shannon capacity.
%, and constructing capacity-achieving codes for models
% that bridge between worst-case and random errors. 
We survey them
briefly here and highlight the relationship to our results.

\mypar{List decoding} List decoding was
introduced in the late 1950s~\cite{elias,wozencraft} and has witnessed
a lot of recent algorithmic work (cf. the survey
\cite{Gur-nowsurvey}).  Under list decoding, the decoder outputs
% is allowed to output 
a small list of messages that must include the correct message. Random
coding arguments demonstrate that there exist binary codes of rate
$1-H(p)-\eps$ which can tolerate $pN$ adversarial errors if the decoder is allowed to output a list of size
$O(1/\eps)$~\cite{elias91,ZP,GHSZ}. 
% If a small amount of auxiliary information can be communicated on a
% noiseless side channel, then it becomes possible to pick the correct
% element from the list with high probability~\cite{Gur-sidechannel}.
The explicit construction of binary list-decodable codes with rate
close to $1-H(p)$, however, remains a major open question. We provide
such codes for the special case of corruptions introduced by space- or time-bounded channels. %
% \focsfull{The}{As  mentioned above, the} %
% model we consider is slightly weaker than the
% standard one, in that we assume that the received word is obtained by
% corrupting a true output of the randomized encoder.

\mypar{Adding Setup---Shared Randomness} Another relaxation 
%that increases codes' power 
is to allow randomized coding strategies where
the sender and receiver share ``secret'' randomness, \emph{hidden from
  the channel}, which is used to pick a particular, deterministic code at random from a
family of codes. Such randomized strategies were called \emph{private codes} in
\cite{langberg04}. Using this secret shared randomness, one can
transmit at rates approaching $1-H(p)$ against adversarial errors (for example, by randomly permuting
the symbols and adding a random
offset~\cite{lipton,langberg04}). Using explicit codes achieving
capacity on the $\bsc_p$~\cite{forney}, one can even get such
randomized codes of rate approaching $1-H(p)$ explicitly (although
getting an explicit construction with $o(n)$ randomness remains an
open problem~\cite{smith07}). A related notion of setup is the
\emph{public key} model of Micali \emph{et al.}~\cite{MPSW}, in which
the sender generates a public key that is known to the
receiver and possibly to the channel. This model only makes sense for
computationally bounded channels, discussed below.

\ifnum\focs=0
Our constructions are the first (for all three models) which achieve
rate $1-H(p)$ with efficient decoding and no setup assumptions.
\fi

\mypar{AVCs: Oblivious, nonuniform errors}
A different
approach to modeling uncertain channels is embodied by the rich
literature on \emph{arbitrarily varying channels} (AVCs), surveyed in \cite{LN-avc-survey}. Despite being extensively investigated in the information theory literature, AVCs have not received much algorithmic attention.

An AVC is specified by a finite state space $\mathcal{S}$ and a family
of memoryless channels $\{W_s: s\in\mathcal{S}\}$. The channel's
behavior is governed by its state, which is allowed to vary
arbitrarily.  The AVC's behavior in a particular execution is
specified by a vector $\vec s=(s_1,...,s_N) \in \mathcal{S}^N$: the
channel applies the operation $W_{s_i}$ to the $i$th bit of the
codeword.  A code for the AVC is required to transmit reliably with
high probability for every sequence $\vec s$, possibly subject
to some state constraint. Thus AVCs model uncertainty via the
\emph{nonuniform} choice of the state vector
$\vec s\in\mathcal{S}^N$. However --- and this is the one of the key
differences that makes the bounded space model more powerful --- the
choice of vector $\vec s$ in an AVC is {\em oblivious} to the
codeword; that is, the
channel cannot look at the codeword to decide the state sequence.

The additive errors channel we consider {\em is} captured
by the AVC framework. Indeed, consider the simple AVC where
$\mathcal{S} = \{0,1\}$ and when in state $s$, the channel adds $s$
mod $2$ to the input bit. With the state constraint $\sum_{i=1}^N s_i
\le p N$ on the state sequence $(s_1,s_2,\dots,s_N)$ of the AVC, this
models  additive errors, where an {\em arbitrary} error vector
$e$ with at most $p$ fraction of $1$'s is added to the codeword by the
channel, but $e$ is chosen {\em obliviously} of the codeword.

Csisz\'{a}r and Narayan determined the capacity of AVCs with state
constraints %for the ``average error criterion''~
\cite{CN88,CN89}. In particular, for the additive case, they showed
that random codes can achieve rate approaching $1-H(p)$ while
correcting any specific error pattern $e$ of weight $pN$ with high
probability.%
\footnote{The AVC literature usually discusses the ``average error
  criterion'', in which the code is deterministic but the message is
  assumed to be uniformly random and unknown to the channel. We prefer
  the ``stochastic encoding'' model, in which  the
   message is chosen adversarially, but the encoder has local random
  coins. The stochastic encoding model strictly stronger than the
  Average error model as long as the
  decoder recovers the encoder's random coins along with message. The
   arguments of Czisz{\'a}r and Narayan \cite{CN88} and Langberg
   \cite{langberg08} also apply to the stronger
  model.}
% A special case of their
% result is the following surprising claim for oblivious additive
% errors: There {\em exist} binary codes $E: \{0,1\}^k \rightarrow
% \{0,1\}^n$ of rate approaching $1-H(p)$ with a deterministic decoding
% rule $D$ such that for {\em every} error vector $e$ of Hamming weight
% at most $pn$, for {\em most} (all but a $\exp(-\Omega(k))$ fraction)
% of messages $m$, $D$ correctly recovers $m$ from $E(m)+e$. 
Note that
codes providing this guarantee {\em cannot} be linear, since the bad
error vectors  are the same for all codewords in a linear code.
\ifnum\focs=0
The decoding rule used in \cite{CN88} to
prove this claim was quite complex, and it was simplified to the more
natural closest codeword rule in
\cite{CN89}. 
\fi
Langberg~\cite{langberg08} 
\ifnum\focs=0 revisited this special case (which he called
an \emph{oblivious channel}) and \fi
gave another proof of the above claim, based on a different random coding argument.

As outlined above, we provide two results for this model. First, we
give a new, more modular existential proof. More importantly, we provide
the first explicit constructions of codes for this model which achieve
the optimal rate $1-H(p)$.

\mypar{Polynomial-time bounded channels} In a different vein,
Lipton~\cite{lipton} considered channels whose behavior can be
described by a polynomial-time algorithm. He showed how a small
amount of secret shared randomness (the seed for a pseudorandom
generator) could be used to communicate at the Shannon capacity over
any polynomial-time channel that introduces a bounded number of
errors. Micali \emph{et al.}~\cite{MPSW} gave a similar result in a
public key model; however, their result relies on efficiently
list-decodable codes, which are only known with sub-optimal rate.
Both results assume the existence of one-way functions and some kind
of setup. On the positive side, in both cases the channel's time bound
need not be known explicitly ahead of time; one gets a trade-off
between the channel's time and its probability of success.

Our list decoding result removes the setup assumptions of
\cite{lipton,MPSW} at the price of imposing a specific polynomial
bound on the channel's running time and relaxing to list-decoding. 
However, our result also implies stronger \emph{unique decoding}
results in the public-key model~\cite{MPSW}.  Specifically, our codes
can be plugged into the construction of Micali \emph{et al.} to get
unique decoding at rates up to the Shannon capacity when the sender
has a public key known to the decoder (and possibly to the channel). The
idea, roughly, is to sign messages before encoding them;
see~\cite{MPSW} for details.  We remark that while they make use of list-decodable codes in the usual sense where all close-by codewords can be found, our weaker notion (where we find a list that includes the original message with high probability) suffices for the application.

Ostrovsky, Pandey and Sahai \cite{OPS07} and
Hemenway and Ostrosky \cite{HO08} considered the construction of
\emph{locally} decodable codes in the presence of computationally
bounded errors assuming some setup (private \cite{OPS07} and public
\cite{HO08} keys, respectively). The techniques used for locally
decodable codes are quite
different from those used in more traditional coding settings; we do not know if the
ideas from our constructions can be used to remove the setup
assumptions from \cite{OPS07,HO08}.

\mypar{Logarithmic-space channels}
Galil \emph{et al.}~\cite{GLYY} considered a slightly weaker model,
logarithmic space, that still captures most physically realizable
channels. They modeled the channel as a finite
automaton with polynomially-many states. Using Nisan's generator for log-space machines~\cite{nisan},
they removed the assumption of one-way functions from Lipton's
construction in the shared randomness model~\cite{lipton}.

In the initial version of this paper, we considered nonuniform
generalization of their model which also generalizes arbitrarily
varying channels.  Our result for polynomial-time bounded channels,
which implies a construction for logarithmic-space channels,
removes the assumption of shared setup in the model of \cite{GLYY}
but achieves only list decoding. This relaxation is necessary for some
parameter ranges, since unique decoding in this model is impossible
when $p>1/4$.

\mypar{Causal channels}
The logarithmic-space channel can also be seen as a restriction of
\emph{online}, or \emph{causal}, channels, recently studied
by Dey, Jaggi and Langberg~\cite{DJL08,LJD09}. These channels make a single-pass through the
codeword, introducing errors as they go. They are not restricted
in either space usage or computation time. It is known that codes for
online channels cannot achieve the Shannon rate; specifically, the
achievable asymptotic rate is at most
$\max(1-4p,0)$ \cite{DJL08} (and in fact slightly tighter bounds were
recently discovered \cite{DeyJLS12}). Our
impossibility result, which shows that the rate of codes for time- or
space-bounded channels is asymptotically 0 for $p>\frac
1 4$, can be seen as a partial extension of the online channels
results  of \cite{DJL08,LJD09,DeyJLS12} to 
computationally-bounded channels, though our proof technique  is quite different.

\section{Statements of Results}
\label{sec:intro-results}
Recall the notion of stochastic codes: A stochastic binary code of
rate $R\in(0,1)$ and block length $N$ is given by an encoding function
$\mathsf{Enc} : \{0,1\}^{RN} \times \{0,1\}^b \rightarrow \{0,1\}^N$
which encodes the $RN$ message bits, together with some additional
random bits ($b$ of them), into an $N$-bit codeword. Here, $N$ and $b$ are integers,
and we assume for simplicity that $RN$ is an integer. 
% Throughout the paper, $N$ will
% denote the block length of the {\em final} code we are interested
% in.

\subsection{Codes for worst-case additive errors}

\noindent {\bf Existential result via list decoding.}
We give a novel construction of stochastic codes for additive errors
by combining {\em linear} list-decodable codes with a certain kind of
authentication code called algebraic manipulation detection (AMD)
codes. Such AMD codes can detect {\em additive corruption} with high
probability, and were defined and constructed for cryptographic
applications in \cite{CDFPW}. 
% The decoder does not have access to the
% randomness $r$ used to authenticate the message $m$. 
The linearity of the
list-decodable code is therefore crucial to make the combination with
AMD codes work. The linearity ensures that the spurious messages
output by the list-decoder are all additive offsets of the true
message and depend only on the error vector (and not on $m,r$).
% By using the existence of binary linear codes that achieve list
% decoding capacity, one can conclude the existence of stochastic codes
% of rate approaching $1-H(p)$ for communication on $\wec_p$. This is a
% simpler proof than the earlier ones~\cite{CN88,CN89}. 
An additional feature of our construction is that even when the
fraction of errors exceeds $p$, the decoder outputs a decoding failure
with high probability (rather than decoding incorrectly). This feature
is important when using these codes as a component in our explicit
construction, mentioned next. 

The formal result is stated below. Details can be found in Section~\ref{sec:ld-implies-avc}.
The notation $\Omega_{p,\eps}$ expresses an asymptotic lower bound in which
$p$ and $\eps$ are held constant.
 
\begin{theorem}
\label{thm:ld-to-avc-intro}
For every $p$, $0 < p < 1/2$ and every $\eps >0$, there {\em exists} a
family of stochastic codes of rate $R\ge 1-H(p)-\eps$ and a
deterministic (exponential time) decoder $\mathsf{Dec} : \{0,1\}^N
\rightarrow \{0,1\}^{RN} \cup \{ \bot \}$ such that for every $m \in
\{0,1\}^{RN}$ and every error vector $e \in \{0,1\}^N$ of Hamming
weight at most $pN$, $\Pr_r\bigl[\mathsf{Dec}\bigl( \mathsf{Enc}(m,r)
+ e \bigr) = m\bigr] \ge 1-2^{-\Omega_{p,\eps}(N)}$.
% and $c(\eps,p)$ is a constant depending only on $\eps$ and $p$. \\
Moreover, when more than a fraction $p$ of errors occur, the decoder
is able to detect this and report a decoding failure ($\bot$) with
probability
at least $1-2^{-\Omega_{p,\eps}(N)}$. \\
Given an \emph{explicit} family of linear binary codes of rate $R$ that can
be efficiently list-decoded from fraction $p$ of errors with list-size
bounded by a polynomial function in $N$, one can construct an \emph{explicit}
stochastic code of rate $R-o(1)$ with the above guarantee along with an
{\em efficient} decoder.
\end{theorem}
% The decoder also recovers the encoder's randomness w.h.p., so we
% also get codes for average error criterion against oblivious
% additive errors via this construction.

\smallskip \noindent {\bf Explicit, efficient codes achieving
  capacity.}  Explicit binary list-decodable codes of optimal rate
are not known, so one cannot use the above connection to construct
\emph{explicit} stochastic codes of rate $\approx 1-H(p)$ for $pN$ additive
errors.  Nevertheless, we give an explicit construction of
capacity-achieving stochastic codes against worst-case additive
errors. The construction is described at a high-level in
Section~\ref{sec:overview} and in further detail in
Section~\ref{sec:efficient-oc-codes}.
\begin{theorem}
\label{thm:additive-intro}
For every $p\in(0,1/2)$, every $\eps>0$, and infinitely many $N$,
there is an explicit, efficient stochastic code of block length $N$
and rate $R \ge 1-H(p)-\eps$ which corrects a $p$ fraction of additive
errors with probability $1-o(1)$. Specifically, there are polynomial
time algorithms $\Enc$ and $\Dec$ such that for
{\em every} message $m \in \{0,1\}^{RN}$ and {\em every} error vector
$e$ of Hamming weight at most $pN$, we have $\Pr_r[\Dec(\Enc(m;r) + e)
=m ] = 1-\exp(-\Omega_{p,\eps}(N/\log^2 N))$.
\end{theorem}

A slight modification of our construction gives codes for the ``average
error criterion,'' in which the code is deterministic but the message
is assumed to be uniformly random and unknown to the channel (see Appendix \ref{app:codes-avg-error}).

\subsection{Unique decoding is impossible for nonoblivious channels
  when $p>\frac 14$}

We exhibit a very simple ``zero space'' (aka memoryless) channel that rules out
achieving any positive rate (i.e., the capacity is zero) when
$p>1/4$. In each position, the channel either leaves the transmitted
bit alone, sets it to 0, or sets it to 1. The channel works by
``pushing'' the transmitted codeword towards a different valid
codeword (selected at random). This simple channel adds at most $n/4$
errors \emph{in expectation}. We can get a channel with a hard bound
on the number of errors by allowing it logarithmic space. Our
impossibility result can be seen as strengthening a result by Dey
\emph{et al.}~\cite{DJL08} for online channels in the special case
where $p>1/4$, though our proof technique, adapted from
Ahlswede~\cite{ahlswede}, is quite different. Appendix~\ref{sec:imposs} contains a
detailed proof.

\begin{theorem}[Unique decoding is impossible for $p\allowbreak >
  \frac 1 4$]
\label{thm:imposs-intro} 
    For every pair of randomized encoding/decoding algorithms
    $\Enc,\Dec$ that make $N$ uses of the channel and use a message
    space whose size tends to infinity with $N$, for every
    $0<\nu<\frac 1 4$, there is an online space-$\lceil\log(N)\rceil$
    channel $\channel_2$ that alters at most $N(\frac 1 4 + \nu)$
    bits and causes a uniformly random message to be incorrectly
    decoded with probability at least $C\cdot \nu$ for an absolute
    constant $C$.
\end{theorem}

\subsection{List-decodable codes for polynomial-time channels}

We next consider a very general noise model. The channel can look at the whole codeword, and effect any error pattern, with the only restriction being that the channel must compute the error pattern in polynomial time given the original codeword as input. In fact, we will even allow non-uniformity, and require that the error pattern be computable by a polynomial size circuit.

\begin{theorem}
  \label{thm:intro-new-time-bounded}
  For all constants $\eps>0$, $p \in (0,1/2)$, and $c\ge 1$, and for
  infinitely many integers $N$, there exists a Monte Carlo
  construction (succeeding with probability $1-N^{-\Omega(1)}$) of a
  stochastic code of block length $N$ and rate $R \ge 1-H(p)-\eps$
  with $N^{O(c)}$ time encoding/list decoding algorithms $(\Enc,\Dec)$
  that have the following property: For all messages $m \in \bit{RN}$,
  and all $pN$-bounded channels $\channel$ that are implementable by a
  size $O(N^c)$ circuit, $\Dec(\channel(\Enc(m;r)))$ outputs a list of
  at most $\mathrm{poly}(1/\eps)$ messages that includes the real
  message $m$ with probability at least $1-N^{-\Omega(1)}$.
\end{theorem}

\section{Simple, nonexplicit codes for worst-case additive errors}
\label{sec:ld-implies-avc}
%
%
% \vnote{What is a good name for this model? Stochastic codes for
%   worst-case noise? Codes for average error? AVC codes? I'm loosely
%   referring to them as AVC codes, but this is not a particularly
%   instructive or appropriate name; stochastic codes may be better
%   but lacks a punch.}
%
In this section, we will demonstrate how to use good {\em linear}
list-decodable codes to get good stochastic codes. The conversion uses
the list-decodable code as a black-box and loses only a negligible
amount in rate. In particular, by using binary {\em linear} codes that
achieve list decoding capacity, we get stochastic codes which achieve
the capacity for additive errors. The linearity of the code is crucial for
this construction. The other ingredient we need for the construction
is an authentication code (called an \emph{algebraic manipulation detection} (AMD) code) that can detect {\em additive
  corruption} with high probability~\cite{CDFPW}.

\subsection{Some coding terminology}

We begin with the definitions relating to list decoding and stochastic
codes for additive errors.

\begin{defn}[List decodable codes]
  For a real $p$, $0 < p < 1$, and an integer $L \ge 1$, a code $C
  \subseteq \Sigma^n$ is said to be $(p,L)$-list decodable if for
  every $y \in \Sigma^n$ there are at most $L$ codewords of $C$ within
  Hamming distance $pn$ from $y$. If for every $y$ the list of $\le L$
  codewords within Hamming distance $pn$ from $y$ can be found in time
  polynomial in $n$, then we say $C$ is {\em efficiently} $(p,L)$-list
  decodable. Note that $(p,1)$-list decodability is equivalent to the
  distance of $C$ being greater than $2pn$. \closedef
\end{defn}
An efficiently $(p,L)$-list decodable code can be used for
communication on the $\adv_p$ channel with the guarantee that the
decoder can always find a list of at most $L$ messages that includes
the correct message.

\begin{defn}[Stochastic codes and their decodability]
\label{def:stochastic}
  A stochastic binary code of rate $R$ and block length $n$ is given
  by an encoding function $\mathsf{Enc} : \{0,1\}^{Rn} \times \{0,1\}^b
  \rightarrow \{0,1\}^n$ which encodes the $Rn$ message bits together
  with some additional random bits into an $n$-bit codeword. 

  Such a code is said to be (efficiently) {\em $p$-decodable with
    probability $1-\delta$} if there is a (deterministic polynomial
  time computable) decoding function $\mathsf{Dec} : \{0,1\}^n
  \rightarrow \{0,1\}^{Rn} \cup \{ \bot \}$ such that for every $m \in
  \{0,1\}^{Rn}$ and every error vector $e \in \{0,1\}^n$ of Hamming
  weight at most $pn$, with probability at least $1-\delta$ over the
  choice of a random string $\omega\in \{0,1\}^b$, we have
  \[ \mathsf{Dec}\bigl( \mathsf{Enc}(m,\omega) + e \bigr) = m \ . \]
  % In this case, we also say that the stochastic code $C$ allows
  % reliable communication over the worst-case channel $\wec_p$ with
  % probability at least $1-\delta$.
\end{defn}

Though we do not require it in the definition, our constructions in
this section of stochastic codes from list-decodable codes will also
have the desirable property that when the number of errors exceeds
$pn$, with high probability the decoder will output a decoding failure
rather than decoding incorrectly.
\subsection{Algebraic manipulation detection (AMD) codes}
The following is not the most general definition of AMD codes from
\cite{CDFPW}, but will suffice for our purposes and is the one we will
use.
\begin{defn}
  Let ${\cal G} = (G_1,G_2,G_3)$ be a triple of abelian groups (whose
  group operations are written additively) and $\delta > 0$ be a
  real. Let $G=G_1 \times G_2 \times G_3$ be the product group (with
  component-wise addition).
  An $({\cal G},\delta)$-algebraic manipulation code, or $({\cal G},\delta)$-AMD code for short, is given by a map $f : G_1 \times G_2 \rightarrow G_3$ with the following property: 
\begin{quote}
  For every $x \in G_1$, and all $\Delta \in G$,
  ~$\Pr_{r \in G_2} \bigl[ D((x,r,f(x,r)) + \Delta) \notin \{x,\bot\}
  \bigr] \le \delta$ \ ,
\end{quote}

where the decoding function $D : G \rightarrow G_1 \cup \{ \bot\}$ is
given by $D((x,r,s)) = x$ if $f(x,r) = s$ and $\bot$ otherwise. The
{\em tag size} of the AMD code is defined as $\log |G_2| + \log |G_3|$
--- it is the number of bits the AMD encoding appends to the
source. \closedef
\end{defn}
Intuitively, the AMD allows one to authenticate $x$ via a signed form
$(x,r,f(x,r))$ so that an adversary who manipulates the signed value
by adding an offset $\Delta$ cannot cause incorrect decoding of some
$x' \neq x$. The following concrete scheme from \cite{CDFPW} achieves
near optimal tag size and we will make use of it.
\begin{theorem}
\label{thm:amd-explicit-const}
  Let $\F$ be a finite field of size $q$ and characteristic $p$, and
  $d$ be a positive integer such that $d+2$ is not divisible by
  $p$. Then the function $f^{(d)}_{\rm AMD} : \F^d \times \F \rightarrow \F$ given by
  $f^{(d)}_{\rm AMD}(x,r) = r^{d+2} + \sum_{i=1}^d x_i r^i$ is a $\bigl({\cal G},
  \frac{d+1}{q}\bigr)$-AMD code with tag size $2 \log q$ where ${\cal
    G} = (\F^d, \F, \F)$.\footnote{Here we mean the additive group of the vector space $\F^d$.}
\end{theorem}
%
%\vspace{-2ex}
\subsection{Combining list decodable and AMD codes}
\label{subsec:ld-avc}
Using a $(p,L)$-list decodable code $C$ of length $n$, for any error
pattern $e$ of weight at most $pn$, we can recover a list of $L$
messages that includes the correct message $m$. We would like to use
the stochastic portion of the encoding to allow us to unambiguously
pick out $m$ from this short list. The key insight is that if $C$ is a
{\em linear} code, then the other (less than $L$) messages in the list
are all fixed offsets of $m$ that depend {\em only} on the error
pattern $e$. So if prior to encoding by the list-decodable code $C$,
the messages are themselves encodings as per a good AMD code, and the
tag portion of the AMD code is good for these fixed $L$ or fewer
offsets, then we can uniquely detect $m$ from the list using the AMD
code. If the tag size of the DMD code is negligible compared to the
message length, then the overall rate is essentially the same as that
of the list-decodable code. Since there exist binary linear
$(p,L)$-list-decodable codes of rate approaching $1-H(p)$ for large
$L$, this gives stochastic codes (in fact, {\em strongly decodable} stochastic
codes) of rate approaching $1-H(p)$ for correcting up to a fraction $p$ of worst-case additive errors.
% The formal claims and proofs follow.
%
\begin{theorem}[Stochastic codes from list decoding and AMD]
\label{thm:ld-to-avc}
  Let $b,d$ be positive integers with $d$ odd and $k = b(d+2)$.  Let
  $C : \F_2^k \rightarrow \F_2^n$ be the encoding function of a binary
  linear $(p,L)$-list decodable code. Let $f^{(d)}_{\rm AMD}$ be the
  function from Theorem~\ref{thm:amd-explicit-const} for the choice
  $\F = \F_{2^b}$. Let $C'$ be the stochastic binary code with encoding map $E : \{0,1\}^{bd} \times \{0,1\}^{b} \rightarrow \{0,1\}^n$ given by \vspace{-1ex}
  \[ \vspace{-1ex} E(m,r) = C\bigl( m,r,f^{(d)}_{\rm AMD}(m,r)\bigr) \ .  \] 
  Then if $\frac{d+1}{2^b} \le \frac{\delta}{L}$, the stochastic code
  $C'$ is strongly $p$-decodable with probability $1-\delta$. If $C$ is
  efficiently $(p,L)$-list decodable, then $C'$ is efficiently (and strongly)
  $p$-decodable with probability $1-\delta$.

Moreover, even when $e$ has weight greater than $pn$, the decoder detects this and outputs $\bot$ (a decoding failure) with probability at least $1-\delta$. 

Note that the rate of $C'$ is $\frac{d}{d+2}$ times the rate of $C$.
\end{theorem}

\begin{proof}
  Fix an error vector $e \in \{0,1\}^n$ and a message $m \in
  \{0,1\}^{bd}$. Suppose we pick a random $r$ and transmit $E(m,r)$,
  so that $y = E(m,r)+e$ was received.  

  The decoding function $D$, on input $y$, first runs the list
  decoding algorithm for $C$ to find a list of $\ell \le L$ messages
  $m'_1,\dots,m'_\ell$ whose encodings are within distance $pn$ of
  $y$. It then decomposes $m'_i$ as $(m_i,r_i,s_i)$ in the obvious
  way. The decoder then checks if there is a unique index $i \in
  \{1,2,\dots,\ell\}$ for which $f^{(d)}_{\rm AMD}(m_i,r_i)= s_i$. If
  so, it outputs $(m_i,r_i)$, otherwise it outputs $\bot$.

  Let us now analyze the above decoder $D$. First consider the case
  when $\wt(e) \le pn$. In this case we want to argue that the decoder
  correctly outputs $(m,r)$ with probability at least $1-\delta$ (over
  the choice of $r$). Note that in this case one of the $m'_i$'s
  equals $(m,r,f^{(d)}_{\rm AMD}(m,r))$, say this happens for $i=1$
  w.l.o.g. Therefore, the condition $f^{(d)}_{\rm AMD}(m_1,r_1)= s_1$
  will be met and we only need to worry about this happening for some
  $i > 1$ also.

  Let $e_i = y - C(m'_i)$ be the associated error vectors for the
  messages $m'_i$. Note that $e_1 = e$. By linearity of $C$, the
  $e_i$'s only depend on $e$; indeed if $c'_1,\dots,c'_\ell$ are all
  the codewords of $C$ within distance $pn$ from $e$, then $e_i = c'_i
  + e$. Let $\Delta_i$ be the pre-image of $c'_i$, i.e., $c'_i =
  C(\Delta_i)$. Therefore we have $m'_i = m'_1 + \Delta_i$ where the
  $\Delta_i$'s only depend on $e$. By the AMD property, for each $i >
  1$, the probability that $f^{(d)}_{\rm AMD}(m_i,r_i) = s_i$ over the
  choice of $r$ is at most $\frac{d+1}{2^b} \le \delta/L$. Thus with
  probability at least $1-\delta$, none of the checks $f^{(d)}_{\rm
    AMD}(m_i,r_i)= s_i$ for $i > 1$ succeed, and the decoder thus
  correctly outputs $m_1 = m$.

  In the case when $\wt(e) > pn$, the same argument shows that the
  check $f^{(d)}_{\rm AMD}(m_i,r_i) = s_i$ passes with probability at
  most $\delta/L$ for each $i$ (including $i=1$). So with probability
  at least $1-\delta$ none of the checks pass, and the decoder outputs
  $\bot$.
\end{proof}

Plugging into the above theorem the existence of binary linear
$(p,O(1/\eps))$-list-decodable codes of rate $1-H(p)-\eps/2$, and
picking $d = 2 \lceil c_0/\eps \rceil +1$ 
%and $b = c_2 \lceil \log(1/(\delta\eps)) \rceil$ 
for some absolute constant $c_0$, we can conclude the following result
on existence of stochastic codes achieving capacity for reliable
communication against additive errors.
\begin{cor}
  For every $p$, $0 < p < 1/2$ and every $\eps >0$, there exists a
  family of stochastic codes of rate at least $1-H(p)-\eps$, which are
  {strongly} $p$-decodable with probability at least $1-2^{-c(\eps,p)
    n}$ where $n$ is the block length and $c(\eps,p)$ is a constant
  depending only on $\eps$ and $p$. \\
  Moreover, when more than a fraction $p$ of errors occur, the code is
  able to detect this and report a decoding failure with probability
  at least $1-2^{-c(\eps,p) n}$.
\end{cor}
\begin{remark}
  For the above construction, if the decoding succeeds, it correctly
  computes in addition to the message $m$ also the randomness $r$ used
  at the encoder. So the construction also gives deterministic codes
  for the ``average error criterion'' where for every error vector,
  all but an exponentially small fraction of messages are communicated
  correctly. See Appendix \ref{app:codes-avg-error} for a discussion
  of codes for this model and their relation to stochastic codes for
  additive errors.
\end{remark}

%%%%OVERVIEW
\section{Overview of Explicit Constructions}
\label{sec:overview}

\mypar{Codes for Additive Errors} Our result is obtained
by combining several ingredients from pseudorandomness and coding
theory. At a high level the idea (introduced by Lipton~\cite{lipton}
in the context of shared randomness) is that if we permute the symbols
of the codewords randomly {\em after} the error pattern is fixed,
then the adversarial error pattern looks random to the
decoder. Therefore, an explicit code $C_{\rm BSC}$ that can achieve
capacity for the binary symmetric channel (such as Forney's
concatenated code~\cite{forney}) can be used to
communicate on $\adv_p$ after the codeword's symbols are randomly
permuted. This allows one to achieve capacity against adversarial
errors when the encoder and decoder share randomness that is unknown
to the adversary causing the errors. But, crucially, this requires the
decoder to know the random permutation  used for encoding.

Our encoder communicates the random permutation (in encoded
form) also as part of the overall codeword, without relying on any
shared randomness, public key, or other ``extra'' information.  The
decoder must be able to figure out the permutation correctly,
based solely on a noisy version of the overall codeword (that encodes the
permutation plus the actual data).  The seed used to pick this random
permutation (plus some extra random seeds needed for the construction)
is encoded by a low rate code that can correct several errors (say, a
Reed-Solomon code) and this information is dispersed into randomly
located blocks of the overall codeword (see
Figure~\ref{fig:prettypic}). The  locations of the control
blocks are picked by a ``sampler'' --- the seed for this sampler is
also part of the {\em control information} along with the seed for the
random permutation.

\begin{figure}[tb]
\noindent  \includegraphics[width=\textwidth]{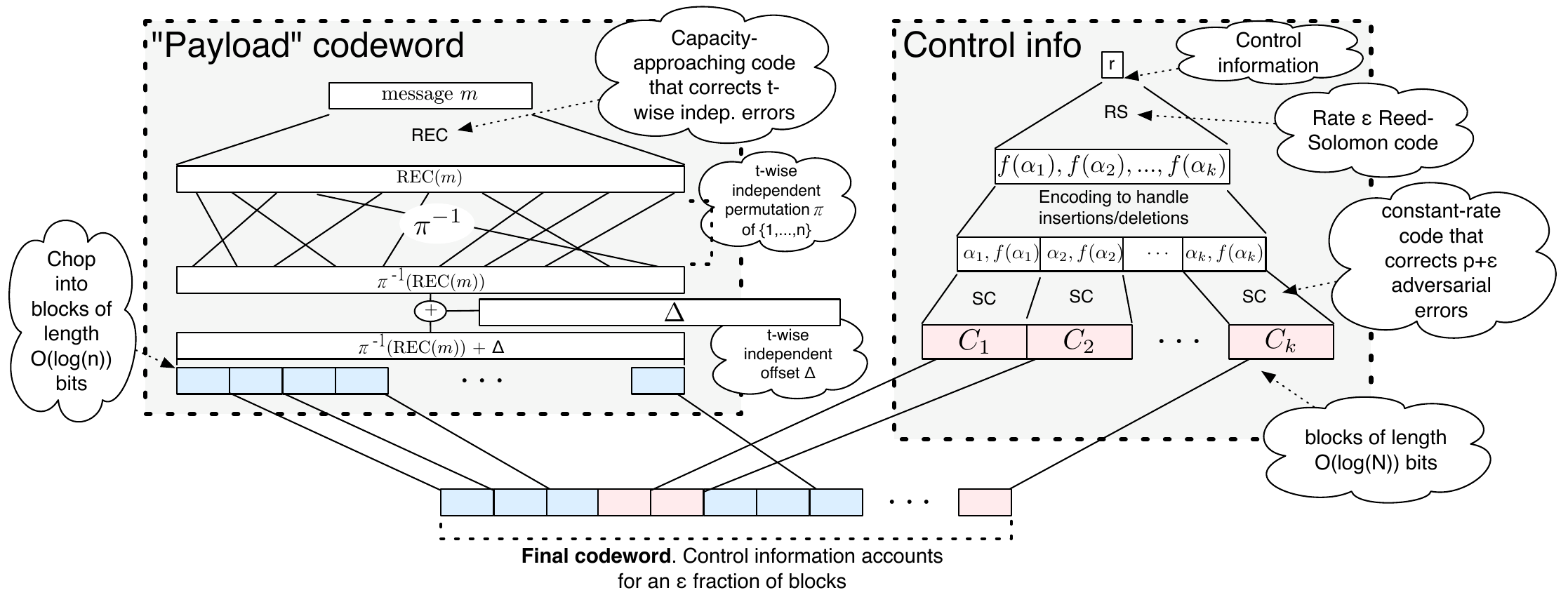}
  \caption{Schematic description of encoder from Algorithm~\ref{alg:enc}.}
\label{fig:prettypic}
\end{figure}

The key challenge is to ensure that the decoder can figure out which
blocks encode the control information, and which blocks consist of
``data'' bits from the codeword of $C_{\rm BSC}$ (the ``payload''
codeword) that encodes the actual message. The control blocks (which
comprise a tiny portion of the overall codeword) are further encoded
by a stochastic code (call it the {\em control code}) that can correct
somewhat more than a fraction $p$, say a fraction $p+\eps$, of
errors. These codes can have any constant rate --- since they encode a
small portion of the message their rate is not so important, so we can
use explicit sub-optimal codes for this purpose.

Together with the random placement of the encoded control blocks, the
control code ensures that a reasonable ($\Omega(\eps)$) fraction of
the control blocks (whose encodings by the control code incur fewer
than $p+\eps$ errors) will be correctly decoded. Moreover, blocks
with too many errors will be flagged as erasures with high
probability. The fraction of correctly recovered control blocks will
be large enough that all the control information can be recovered by
decoding the Reed-Solomon code used to encode the control information
into these blocks. This recovers the permutation used to scramble the
symbols of the concatenated codeword. The decoder can then unscramble
the symbols in the message blocks and run the standard algorithm for
the concatenated code to recover the message.

One pitfall in the above approach is that message blocks could potentially get
mistaken for corrupted control blocks and get decoded as erroneous
control information that leads the whole algorithm astray. To prevent
this, in addition to scrambling the symbols of the message blocks by a
(pseudo)random permutation, we also add a pseudorandom offset (which
is nearly $t$-wise independent for some $t$ much larger than the
length of the blocks). This will ensure that with high probability
each message block will be very far from every codeword and therefore
will not be mistaken for a control block.

An important issue we have glossed over is that a uniformly random
permutation of the $n$ bits of the payload codeword would take $\Omega(n
\log n)$ bits to specify. This would make the control information 
too big compared to the message length; we need it to be a
tiny fraction of the message length. We therefore use
almost $t$-wise independent permutations for $t \approx \eps n/\log
n$. Such permutations can be sampled with $\approx \eps n$ random
bits. We then make use of the fact that $C_{\rm BSC}$ enables reliable
decoding even when the error locations have such limited independence
instead of being a uniformly random subset of all possible locations~\cite{smith07}.

\mypar{Extending the Construction to log-space and poly-time channels}
%
%\anote{This section partly  rewritten.}
%
The construction for additive channels does not work against more
powerful channels for (at least) two reasons: 
\begin{enumerate}
\itemsep=1ex
\item[(i)]
A more powerful channel
may inject a large number of correctly formatted control blocks into
the transmitted word (recall, each of the blocks is quite small). Even
if the real control blocks are uncorrupted, the decoder will have
trouble determining which of the correct-looking control blocks is in
fact legitimate. 
\item[(ii)] Since the channel can decide the errors after seeing parts of the codeword, it 
may be able to learn which blocks of the codeword contain the control information and
concentrate errors on those blocks. Similarly, we have to ensure that the channel does not learn 
about the permutation used to scramble the payload codeword and thus cause a bad error pattern that cannot be decoded by the standard decoder for the concatenated code. 
\end{enumerate}

The first obstacle is the easier one to get around, and we do so by using list-decoding: although the
channel may inject spurious possibilities for the control information,
the total number of such spurious candidates will be bounded. This ensures that after list decoding, provided at least a small fraction of the true control blocks do not incur too many errors, the
list of candidates will include the correct control information with
high probability.

To overcome the second obstacle, we make sure, using appropriate
pseudorandom generators and employing a ``hybrid'' argument, that the
encoding of the message is \emph{indistinguishable from a random string} by a
channel limited to a given time bound $\T$, even when the channel has
knowledge of the message and certain parts of the control information.
%(For concreteness, let us focus on the online log-space case in the
%following discussion.) 
%
One then uses this to ensure that the \emph{distribution of errors}
caused by a time-$\T$ channel on the codeword is indistinguishable in
polynomial time from the distribution caused by the same channel on a
uniformly random string. 
The latter distribution is
independent of the codeword. 
If these error distributions were in fact
{\em statistically close} (and not just close w.r.t. time-$\T$ tests), successful decoding under oblivious errors would also imply
successful decoding under the error distribution caused by the
time bounded channel. To show that closeness
w.r.t. time-$\T$ tests does indeed suffice, we need to consider
each of the conditions for correct decoding separately.

The condition that enough control blocks have at most a fraction
$p+\eps$ of errors can be checked in polynomial time given nonuniform
advice, namely the locations of the control blocks. We use this together
with the above indistinguishability to prove that enough control
blocks are correctly list-decoded, and thus the correct control
information is among the candidates obtained by list decoding the
control code.

The next step is to show that the payload codeword is correctly
decoded given knowledge of the correct control information. The idea
is that there is a set of error patterns such that:
\begin{inparaenum}\item 
  membership in the set can be checked in linear time, \item the set has
  high probability under any oblivious error distribution, and \item
  any error pattern in the set is correctly decoded with high
  probability by the concatenated code.
\end{inparaenum}
Given these properties, one can show that if the concatenated code
errs with noticeable probability on the actual error distribution, one
can build a low-complexity distinguisher for the error distributions, thus
contradicting their computational indistinguishability. 

\begin{remark}
  An earlier version of this paper had a more complicated argument for
  online logspace channels, replacing one of the random components of
  the construction with Nisan's explicit pseudorandom generator for
  logspace. Nisan's generator only ensures that the error distribution
  caused by the channel is indistinguishable from oblivious errors by
  {\em online} space-bounded machines. Therefore, in the above
  argument, in order to arrive at a contradiction we need to build a
  distinguisher than runs in online logspace. However, the
  unscrambling of the error vector (according to the permutation that
  was applied to the payload codeword) cannot be done in an online
  fashion. So we had to resort to an indirect argument based on
  showing limited independence of certain events related to the
  payload decoding. As pointed out by a reviewer, since the final
  construction is anyway randomized, in the current version we simply
  use a random construction of pseudorandom generators for polynomial
  size circuits to build codes resilient to polynomial time bounded
  channels (and hence also logspace channels).
\end{remark}

%%%%OVERVIEW

\section{Explicit Codes of Optimal Rate for Additive Errors}
\label{sec:efficient-oc-codes}

This section describes our construction of codes for
additive errors. 

\subsection{Ingredients}
\label{sec:ingre-brief}
Our construction uses a number of tools from coding theory and
pseudorandomness. \full{These are described in detail in
Appendix~\ref{sec:ingredients}. Briefly, we use:}
\begin{itemize}
\itemsep=1ex
\item A constant-rate explicit stochastic code
  $\SC:\bit{b}\times\bit{b}\to\bit{c_ob}$, defined on blocks of length
  $c_0b=\Theta(\log N)$, that is efficiently decodable with
  probability $1-c_1/N$ from a fraction $p+O(\eps)$ of \emph{additive}
  errors.  decodable with probability $1-c_1/N$. These codes are
  obtained via Theorem~\ref{thm:ld-to-avc-intro}\full{ (see
  Proposition~\ref{prop:const-rate-AVC} in the appendix)}.

\item A rate $O(\eps)$ Reed-Solomon code $\RS$ which encodes a message as the evaluation of a polynomial at points $\alpha_1,...,\alpha_\ell$ in such a way that an efficient algorithm $\text{\sc RS-Decode}$ can efficiently recover the message given at most  $\eps \ell/4$ correct symbols and at most $\eps/24$ incorrect ones. 

\item A randomness-efficient {\em sampler} $\Samp : \{0,1\}^\sigma \rightarrow [N]^\ell$ , such that for any subset $B\subseteq[N]$ of size at least $\mu N$, the output set of the sampler intersects with $B$ in roughly a $\mu$ fraction of its size, that is
$|\Samp(s) \cap B| \approx \mu |\Samp(s)|$, with high probability over
$s\in\bit{\sigma}$ \focsfull{(\emph{e.g.}, due to Vadhan~\cite{vadhan04}).}{. We use an expander-based construction from Vadhan~\cite{vadhan04}.}

\item A generator $\knr : \{0,1\}^\sigma\rightarrow \spcbound_n$ for an (almost) $t$-wise independent family of permutations of the set $\{1,...,n\}$, that uses a seed of $\sigma=O(t\log n)$ random bits~(Kaplan, Naor, and Reingold~\cite{KNR}).

\item A generator $\twise_t: \{0,1\}^{\sigma} \rightarrow \{0,1\}^n$
 for a  $t$-wise independent distribution of bit strings of length $n$, that uses a seed of $\sigma=O(t\log n)$ random bits.

\item An explicit efficiently decodable, rate $R=1-H(p)-O(\eps)$ code
  $\REC : \{0,1\}^{Rn} \rightarrow \{0,1\}^n$ that can correct a $p$
  fraction of $t$-wise independent errors, that is: for every message
  $m \in \{0,1\}^{Rn}$, and every error vector $e \in \{0,1\}^n$ of
  Hamming weight at most $pn$, we have $\text{\sc REC-Decode}(\REC(m)+\pi(e)) = m$ with
  probability at least $1-2^{-\Omega(\eps^2 t)}$ over the choice of a
  permutation $\pi \in_R \text{range}(\knr)$. (Here $\pi(e)$ denotes
  the permuted vector: $\pi(e)_i = e_{\pi(i)}$.) A standard family of
  concatenated codes satisfies this property (see, \emph{e.g.}, \cite{smith07}).
\end{itemize}

\subsection{Analysis}

\full{The following (Theorem~\ref{thm:additive-intro}, restated) is our
result on explicit construction of capacity-achieving codes for
additive errors.
\begin{theorem}\label{thm:efficient-oc-code}
  For every $p\in(0,1/2)$, and every $\eps>0$, the functions {\sc
  Encode}, {\sc Decode} (Algorithms~\ref{alg:enc}
  and~\ref{alg:dec}) form an explicit, efficiently encodable and
  decodable stochastic code with rate $R=1-H(p)-\eps$ such that for
  every $m \in \bit{RN}$ and error vector $e \in \bit{N}$ of Hamming
  weight at most $pN$, we have $\Pr_\omega \bigl[ \text{\sc Decode}(\text{\sc Encode}(m;\omega) +
  e) =m \bigr] \ge 1-\exp(-\Omega(\eps^2 N / \log^2 N)))$, where $N$ is the
  block length of the code.
  % \vnote{Changed the $\log^2 N$ to $\log N$ (and made $\eps^2$ and
  %   $\eps^3$). Should be checked if this is correct,}
\end{theorem}
With all the ingredients described in Section~\ref{sec:ingredients} in
place, we can describe and analyze the code of
Theorem~\ref{thm:efficient-oc-code}.
}
The encoding algorithm is given
in Algorithm~\ref{alg:enc} (page \pageref{alg:enc}). The corresponding
decoder is given in Algorithm~\ref{alg:dec} (page
\pageref{alg:dec}). Also, a schematic illustration of the encoding
is in Figure~\ref{fig:prettypic}.  The reader might find it useful to
keep in mind the high level description from
Section~\ref{sec:overview} when reading the formal
description. %\anote{Should try to make alg. more readable.}

%%% ALGORITHM DESCRIPTION
\theoremstyle{definition}
\newtheorem{myalgorithm}{Algorithm}
\newcommand{\mycaption}[1]{#1}
\newcommand{\fw}[1]{\makebox[1.2in][l]{#1}} % the point of this to make comments that appear on a line with other text all have the same offset from the margin. Just wrap the preceding text in \fw{...}

 \begin{figure}[p]
\fbox{\parbox{0.98\textwidth}{
\begin{myalgorithm}\label{alg:enc}\small
\mycaption{{\sc Encode:} On input parameters $N,p,\eps$ (with $p+\eps < 1/2$), and message
$m\in \bit{R\cdot N}$, 
where \mbox{$R = 1-H(p)-O(\eps)$}. }

\begin{algorithmic}[1]\setlength{\itemsep}{4pt}
 
 \STATE \fw{$\Lambda \gets 2c_0$}  \comment {Here $c_0$
   is the expansion of the 
   stochastic code from Theorem~\ref{thm:ld-to-avc-intro}
  that can correct a fraction $p+\eps$  of errors.}
  %(assuming $\eps< \frac{p-\frac 1 2}{2}$).
%
\\%\STATE
 \fw{$n\gets \frac{N}{\Lambda \log N}$} \comment{The final codeword consists of $n$ {\em blocks} of length $\Lambda \log N$. %(assume for convenience that $\log(N)\in \mathbb{N}$).
}
\\%\STATE
 \fw{$\ell \gets 24\eps N/\log N$} \comment{The control codeword is $\ell$ blocks long.}
\\%\STATE
 $n' \gets n-\ell$ and  $N'\gets n'\cdot(\Lambda \log N) $
\scomment{The payload codeword is $n'$ blocks long (i.e. $N'$ bits).}

   \boldcomment{Phase 1: Generate control information}

\STATE \textbf{Select seeds} $s_\pi,s_\Delta, s_\ctrlset$ uniformly in $\bit{\epsilon^2 N}$. 

% \STATE \fw{Select $s_\pi \gets_R \bit{\epsilon^2 N}$. }
%  \comment{$s_\pi$
% is a seed for picking a permutation of $[N']$ from an almost $t$-wise independent 
% family as per Proposition~\ref{prop:indep-perms}, where $t=\Omega(\epsilon^2 N /\log
% N)$.}

% \STATE \fw{Select $s_\Delta \gets_R \bit{\epsilon^2 N}$. }
% %
%  \comment{$s_\Delta$ is a
% seed for picking a $t'$-wise independent string $\Delta$, where $t'=\Omega(\epsilon^2 N/\log N)$ as per Proposition~\ref{prop:twise}.}

% \STATE \fw{Select $s_\ctrlset \gets_R \bit{\epsilon^2 N}$. }
% %
%  \comment{$s_\ctrlset$ is a
% seed for sampling a pseudorandom subset $\ctrlset\subset
% [n]=[n'+\ell]$ of size $\ell$ as per Proposition~\ref{prop:sampler}.}

\focsfull{Set the \emph{control information} $\omega$ to be the concatenation $ (s_\pi,s_\Delta,s_\ctrlset)$.}{\STATE \fw{$\omega \gets (s_\pi,s_\Delta,s_\ctrlset)$ }
 \comment{Total length $|\omega|= 3\epsilon^2 N$.}}

 \boldcomment{Phase 2: Encode control information}

\STATE 
\textbf{Encode $\omega$ with a Reed-Solomon code} $\RS$ to get symbols $(a_1,...,a_\ell)$.\\
\wcomment{$\RS$ is a rate $\frac\eps8$ Reed-Solomon code of length $24
  \epsilon N = \frac{8}{\eps}\cdot |\omega|$ bits which evaluates polynomials at
points  $(\alpha_1,\dots,\alpha_\ell)$ in a field $\F$ of size 
$\approx N$.}

\STATE \textbf{Encode each symbol together with its evaluation point}: For
$i=1,...,\ell$, do
\begin{itemize}\item 
  $A_i \gets (\alpha_i,a_i)$

  \full{\comment{We add location information to each RS symbol to
      handle insertions and deletions.} } 
%\STATE
\item 
  \label{state:sc-encoding} $C_i \gets \SC(A_i,r_i)$, where $r_i$ is
  random of length $2\log N$ bits.
\\
  \comment{\focsfull{$\SC$}{$\SC=\SC_{2 \log N,p+\eps} :\bit{2\log N}
      \times \bit{2 \log N} \to\bit{\Lambda \log N}$} \focsfull{corrects additive errors with high
      probability}{is a stochastic
    code that can correct a fraction $(p+\eps)$ of additive
      errors with probability $1-c_1/N^2 > 1 - 1/N$ as per
      Proposition~\ref{prop:const-rate-AVC}}.}
  % \wcomment{The control information $\omega$ is thus encoded by a
  % concatenated code with an outer Reed-Solomon code and inner code
  % $\SC$.}
 \end{itemize}

\boldcomment{Phase 3: Generate the payload codeword}

\STATE \textbf{Encode $m$ using a code that corrects 
\emph{random} errors}:
\begin{itemize}\item 
  $P\gets \REC(m)$, \comment{ $\REC:\bit{R' N'}\to \bit{N'}$ is a code
    that corrects a $p+25\Lambda\eps$ fraction
    of \emph{$t$-wise independent} errors\full{, as per
      Proposition~\ref{prop:concat-codes}}. Here $R' =
    \tfrac{RN}{N'}$.}
\end{itemize}

\STATE \textbf{Expand the seeds} $s_\ctrlset, s_\Delta, s_\pi$ to get
a set $\ctrlset=\Samp(s_\ctrlset)$, offset $\Delta=\twise(s_\Delta)$, and permutation  $\pi=\knr(s_\pi)$.

\STATE \textbf{Scramble the payload codeword:}

\begin{itemize}
  \item $\pi^{-1}(P)\gets (\text{bits of $P$ permuted according to
    $\pi^{-1}$})$

  \item $Q \gets \pi^{-1}(P)\oplus \Delta$
  \item  Cut $Q$ into $n'$ blocks $B_1,...B_{n'}$ of length $\Lambda \log N$ bits.
\end{itemize}
\boldcomment{Phase 4: Interleave blocks of payload codeword and control
codeword}

\STATE \textbf{Interleave} control blocks $C_1,...,C_\ell$ with
payload blocks $B_1,...,B_{n'}$,
using control blocks in positions from $\ctrlset$ and payload blocks in
remaining positions.

\end{algorithmic}
\end{myalgorithm}
}} % matches \fbox{\parbox...{
\end{figure}

 \begin{figure}[t!]
\fbox{\parbox{0.98\textwidth}{
\begin{myalgorithm}\label{alg:dec}
\mycaption{{\sc Decode}: On input $x$ of  length $N$:
%the length output by \Enc.
}

% \noindent
% \wcomment{The decoder's pseudocode is annotated with statements about  performance. These claims assume that $x=\Enc(m;\omega,r_1,\dots,r_\ell)+e$ where $e$ contains at most a fraction $p$ of ones and the random string $(\omega; r_1,r_2,\dots,r_\ell)$ is uniform and independent of the pair $(m,e)$. }
\begin{algorithmic}[1] \setlength{\itemsep}{4pt}
\STATE Cut $x$ into $n'+\ell$ blocks $x_1,...,x_{n'+\ell}$ of length
$\Lambda \log(n)$ each.  

\STATE \textbf{Attempt to decode control blocks}: For
$i=1,...,n'+\ell$, do
\begin{itemize}
\item $\tilde F_i\gets \text{\sc SC-Decode}(x_i)$. \\
\comment{With high prob, non-control blocks are rejected (Lemma~\ref{lem:payloadblocks}), and control blocks are either correctly decoded or discarded (Lemma~\ref{lem:controlblocks}).}

\item If $\tilde F_i\neq \perp$, then
parse $\tilde F_i$ as $(\tilde\alpha_i,\tilde a_i)$, where
$\tilde\alpha_i, \tilde a_i \in \F$.
\end{itemize}

% \FOR{$i\gets 1 $ to $n'+\ell$}

%     \STATE $\tilde F_i\gets \text{\sc SC-Decode}(x_i)$. 
% \wcomment{Attempt to decode control blocks}
% \\
% \wcomment{With high prob, non-control blocks are rejected (Lemma~\ref{lem:payloadblocks}), and control blocks are either correctly decoded or discarded (Lemma~\ref{lem:controlblocks}).}
%     \IF{$\tilde F_i\neq \perp$}

%     \STATE Parse $\tilde F_i$ as $(\tilde\alpha_i,\tilde a_i)$, where
% $\tilde\alpha_i, \tilde a_i \in \F_N$.

% \ENDIF

% \vspace{-4pt}
% \ENDFOR 

\STATE $(\tilde s_\ctrlset,\tilde s_{\Delta}, \tilde s_\pi)\gets\text{\sc
RS-Decode}\Big(\text{pairs $(\tilde\alpha_i,\tilde a_i)$ output
above}\Big)$.  
\\
 \comment{Control information is recovered w.h.p.\full{ (Lemma~\ref{lem:controlinfo}).}}

\STATE Expand the seeds $\tilde s_\ctrlset, \tilde s_\Delta, \tilde s_\pi$ to get
set $\tilde \ctrlset$, offset $\tilde \Delta$, and permutation  $\tilde \pi$. 
% $\tilde \ctrlset\gets \Samp(\tilde s_\ctrlset)$, 
% \\
% $\tilde\Delta \gets \twise(\tilde s_\Delta)$
% \\
%  $\tilde
% \pi \gets \knr(\tilde s_{\pi})$

\STATE $\tilde Q \gets $ concatenation of blocks $x_i$ not in
$\tilde \ctrlset$
\full{\\
\comment{Fraction of errors in $\tilde Q$ is at
    most $p+O(\eps)$.}
}

\STATE $\tilde P \gets \pi(\tilde Q \oplus \tilde \Delta)$ 
 \comment{If control info is correct, then errors in $\tilde P$ are almost $t$-wise
     independent.}

\STATE $\tilde m \gets \text{\sc REC-Decode}(\tilde P)$ 

\full{\comment{Run the decoder from Proposition~\ref{prop:concat-codes}.}}

\end{algorithmic}
\end{myalgorithm}
}} % matches \fbox{\parbox...{
\end{figure}
%%% ALGORITHM DESCRIPTION

\full{\mypar{Starting the Proof of
    Theorem~\ref{thm:efficient-oc-code}}} First, note that the rate
$R$ of the overall code approaches the Shannon bound: $R$ is almost
equal to the rate $R'$ of the code $\REC$ used to encode the actual
message bits $m$, since the encoded control information has length
$O(\eps N)$. The code $\REC$ needs to correct a fraction $p+25\Lambda
\eps$ of $t$-wise independent errors, so we can pick $R' \ge
1-H(p)-O(\eps)$. Now the rate $R = \frac{R' N'}{N} = R'(1-24\Lambda
\eps) \ge 1-H(p)-O(\eps)$ (for small enough $\eps > 0$).

We now turn to the analysis of the decoder.  Fix a message $m \in
\bit{R\cdot N}$ and an error vector $e \in \bit{N}$ with Hamming
weight at most $pN$. Suppose that we run $\Enc$ on $m$ and coins
$\omega$ chosen {\em independently} of the pair $m,e$, and let $x =
\Enc(m;\omega)+e$. The decoder parses $x$ into blocks
$x_1,...,x_{n'+\ell}$ of length $\Lambda \log N$, corresponding to the
blocks output by the encoder.

The \focsfull{three}{four} lemmas below, proved in \focsfull{the full
  version,}{Section~\ref{sec:proofs-from-main},} show that the decoder
recovers the control information correctly with high
probability. We then \focsfull{sketch why the payload message is correctly
recovered.}{show that the payload message is correctly
recovered. The proof of the theorem is completed in
Section~\ref{sec:pf-of-main-thm}.}

The lemmas illuminate the roles of the main pseudorandom objects in
the construction. 
First, the sampler seed is used to ensure that errors are
not concentrated on the control blocks%
\focsfull{%
. We say a sampled set $\ctrlset$ is {\em good} for error vector $e$ if the fraction
  of control blocks with relative error rate at most $p+\eps$ is at
  least $\frac\eps2$. 
}{%PART 2 of focsfull
, as captured in the next lemma:
\begin{defn}[Good sampler seeds]\label{def:goodsampler}
  A sampled set $\ctrlset$ is {\em good} for error vector $e$ if the fraction
  of control blocks with relative error rate at most $p+\eps$ is at
  least $\frac\eps2$. \closedef
\end{defn}}

\begin{lemma}[Good sampler lemma]
  \label{lem:sampling} For any error vector $e$ of relative weight at
  most $p$, with probability at least $1-\exp(-\Omega(\eps^3 N / \log
  N)$ over the choice of sampler seed $s_\ctrlset$, the set $\ctrlset$ is \emph{good} for
  $e$.
\end{lemma}

Given a good sampler seed, the properties of the stochastic code $\SC$
guarantee that many control blocks are correctly
interpreted. Specifically:
\begin{lemma}[Control blocks lemma]\label{lem:controlblocks}
  For all $e,\ctrlset$ such that $\ctrlset$ is good for $e$, with probability at
  least $1-\exp(-\Omega(\eps^3N/\log N))$ over the random coins
  $(r_1,r_2,\dots,r_\ell)$ used by the $\ell$ $\SC$ encodings, we have:
  \begin{inparaenum}
    \item[(i)] The number of
   control blocks correctly decoded by {\sc SC-Decode} is at least $\frac{\eps\ell}{4}$, and
  \item[(ii)] The number of
    {\em erroneously} decoded control blocks is less than
    $\frac{\eps\ell}{24}$.
   \full{\\}
    (By erroneously decoded, we mean that {\sc SC-Decode} outputs neither
    $\perp$ nor the correct message.)
  \end{inparaenum}
\end{lemma}

The offset $\Delta$ is then used to ensure that payload blocks are not
mistaken for control blocks:

\begin{lemma}[Payload blocks lemma]\label{lem:payloadblocks}
  For all $m$, $e$, $s_\ctrlset$, $s_\pi$, with probability at least
  $1-2^{-\Omega(\eps^2 N / \log^2 N))}$ over the offset
  seed $s_\Delta$, the number of {\em payload} blocks incorrectly
  accepted as control blocks by $\text{\sc SC-Decode}$ is less than
  $\frac{\eps\ell}{24}$.
 \end{lemma}

The two previous lemmas imply that the Reed-Solomon decoder will, with
high probability, be able to recover the control information.
\full{Specifically:
\begin{lemma}[Control Information Lemma]\label{lem:controlinfo}
  For any $m$ and $e$, with probability $1-2^{-\Omega(\eps^2 N/\log^2 N)}$ over the choice of the control information
  and the coins of $\SC$, the control information is correctly
  recovered, that is $ (\tilde s_\ctrlset, \tilde s_\Delta, \tilde s_\pi)=(s_\ctrlset,  s_\Delta, s_\pi)$.
\end{lemma}
}

\ifnum\focs=1 %MATCHES NOPROOFS
It remains to analyze the final decoding process. First, suppose that
the correct
control information $\omega=(s_\pi,s_\Delta,s_\ctrlset)$ is handed directly
to the decoder--- \emph{i.e.}, assume we are in the ``shared randomness'' setting.
Fix $m$, $e$, and $s_\ctrlset$, and let
$e_Q$ be the restriction of the error $e$ to the payload codeword. The relative weight of $e_Q$ is at most
$\frac{pN}{N'} \leq p(1+25\Lambda\eps)$ for sufficiently
small $\eps$.
Consider the string $\tilde P $ that is input the the $\REC$
decoder. We can write $\tilde P = \tilde\pi(\tilde Q \oplus \tilde
\Delta) = \pi(Q\oplus e_Q \oplus \Delta)  = P\oplus \pi(e_Q)$.
Since $s_\pi$ is selected independently from $\ctrlset$,
the permutation $\pi$ is independent of the payload error $e_Q$.
and so the input to $\REC$ is corrupted by at most
  $p(1+25\Lambda\eps)N'$ errors which are $t$-wise independent, for
  appropriate $t$. By the properties of $\REC$,
  with probability at least $%1-e^{-\Omega(\eps^2 t)} =
  1-e^{-\Omega(\eps^4 N/\log N)}$, the message $m$ is correctly
  recovered by $\text{\sc Decode}$.

  The real error probability is slightly higher since we must
  condition on correct recovery of the control information; see the
  full version for details.

% In the actual decoding, the control information is not handed
%   directly to the decoder. Let $\tilde{\omega}$ be the candidate control information recovered by the decoder (in Step 8 of the algorithm). 
% The above suite of lemmas (Lemmas \ref{lem:sampling},
%   \ref{lem:controlblocks}, \ref{lem:payloadblocks}, and
%   \ref{lem:controlinfo}) show that the control information is
%   correctly recovered, i.e., $\tilde{\omega}=\omega$, with probability at least 
%  $\exp(-\Omega(\eps^2 N / \log^2
%   N)))$.

% The overall probability of success is given by 
% \[ \Pr_\omega [ \mbox{payload decoding succeeds with control information  $\tilde{\omega}$} ] \]
% which is at least 
% \begin{align*}
% \Pr_\omega [ \tilde{\omega} = \omega ~ \wedge ~ \mbox{payload decoding
%   succeeds with control information  $\tilde{\omega}$} ] \\
% = \Pr_\omega [ \tilde{\omega} = \omega ~ \wedge ~ \mbox{payload decoding succeeds with control information  $\omega$} ] \\
% \ge  1 - \Pr_\omega [ \tilde{\omega} \neq \omega ] 
%  -  \Pr_\omega [ \mbox{payload decoding succeeds given $\omega$} ] \\
% \ge 1 - \exp(-\Omega(\eps^2 N / \log^2 N)) - \exp(-\Omega(\eps^4 N/\log N)) \ .
% \end{align*}
%  Because $\eps$ is a constant relative to $\log N$, it is the former
%   probability that dominates. This completes the analysis of the
%   decoder and the proof of Theorem~\ref{thm:efficient-oc-code}.

\else %MATCHES NOPROOFS
\begin{remark}
  It would be interesting to achieve an error probability of
  $2^{-\Omega_\eps(N)}$, i.e., a positive ``error exponent,'' in
  Theorem~\ref{thm:efficient-oc-code} instead of the
  $2^{-\Omega_\eps(N/\log^2 N)}$ bound we get. A more careful analysis
  (perhaps one that works with {\em almost} $t'$-wise independent
  offset $\Delta$) can probably improve our error probability to
  $2^{-\Omega_\eps(N/\log N)}$, but going further using our approach
  seems difficult.  The existential result due to Csisz\'{a}r and
  Narayan~\cite{CN88} achieves a positive error exponent for all rates
  less than capacity, as does our existence proof using list decoding in Section~\ref{subsec:ld-avc}.
\end{remark}

\begin{remark}
  A slight modification of our construction give codes for the
  ``average error criterion,'' in which the code is deterministic but
  the message is assumed to be unknown to the channel and the goal is
  to ensure that for every error vector most messages are correctly
  decoded; see Theorem~\ref{thm:efficient-avc-code} in
  Appendix~\ref{app:codes-avg-error}.
\end{remark}

\iffalse
In Appendix~\ref{app:explicit-avc}, we apply
Observation~\ref{obs:ssc-avc} to a modification of the above
stochastic code (that makes it strongly $p$-decodable) and conclude
the following:
%
\begin{theorem}
\label{thm:efficient-avc-code}
For every $p\in(0,1/2)$, and every $\eps>0$, there is an explicit
family of binary codes of rate at least $1-H(p)-\eps$ that are
efficiently $p$-decodable with average error $O(1/N)$ where $N$ is the
block length of the code.
\end{theorem}
\fi

\subsection{Proofs of Lemmas used in
  Theorem~\ref{thm:efficient-oc-code}}
\label{sec:proofs-from-main}

   \begin{proof}[of Lemma~\ref{lem:sampling}]
     Let $B\subset [n]=[n'+\ell]$ be the set of blocks that contain a
     $(p+\eps)$ or smaller fraction of errors.  We first prove that
     $B$ must occupy at least an $\eps$ fraction of total number of
     blocks: to see why, let $\gamma$ be the proportion of blocks
     which have error rate at most $(p+\eps)$. The total fraction of
     errors in $x$ is then at least $(1-\gamma)(p+\eps)$. Since this fraction is at most $p$ by assumption, we must have $1-\gamma \le p/(p+\eps)$. So $\gamma \ge \eps/(p+\eps) > \eps$.

     Next, we show that the number of {\em control blocks} that have error
     rate at most $p+\eps$ cannot be too small.  The error $e$ is
     fixed before the encoding algorithm is run, and so the sampler
     seed $s_\ctrlset$ is chosen independently of the set $B$. Thus, the
     fraction of control blocks in $B$ will be roughly
     $\eps$. Specifically, we can apply Proposition~\ref{prop:sampler}
     with $\mu=\eps$ (since $B$ occupies at least an $\eps$ fraction
     of the set of blocks), $\theta=\eps/2$ and $\sigma=\eps^2N$. We
     get that the error probability $\gamma$ is
     $\exp(-\Omega(\theta^2\ell)) = \exp(-\Omega(\eps^3N/\log
     N)$. (Note that for constant $\eps$, the seed length $\sigma =
     \eps^2N \gg \log N + \ell \log(1/\eps)$ is large enough for the
     proposition to apply.)
\end{proof}

\begin{proof}[of Lemma~\ref{lem:controlblocks}] 
  Fix $e$ and the sampled set $\ctrlset$ which is good for $e$.  Consider a
  particular received block $x_i$ that corresponds to control block
  $j$, that is, $x_i = C_j + e_i$.  The key observation is that the
  error vector $e_i$ depends on $e$ and the sampler seed $\ctrlset$, but it
  is {\em independent} of the randomness used by $\SC$ to generate
  $C_j$. Given this observation, we can apply
  Proposition~\ref{prop:const-rate-AVC} directly:

  \begin{enumerate}
  \item[(a)] If block $i$ has error rate at most $ p+\eps$, then $\text{\sc
      SC-Decode}$ decodes correctly with probability at least
    $1-c_1/N^2 \ge 1- 1/N$ over the coins of $\SC$.
  \item[(b)] If block $i$ has error rate more than $p+\eps$, then
    $\text{\sc SC-Decode}$ outputs $\perp$ with probability at least $1-c_1/N^2 \ge 1-1/N$
    over the coins of $\SC$.
  \end{enumerate}

  Note that in both statements (a) and (b), the probability need only be
  taken over the coins of $\SC$.

  Consider $\bf Y$, the the number of control blocks that either (i) have ``low'' error rate ($\leq
  p+\eps$) yet are not correctly decoded, or (ii) have high error
  rate, and are not decoded as $\perp$. Because statements
  (a) and (b) above depend only on the coins of $\SC$, and these coins
  are chosen independently in each block, the variable  $\bf Y$
  is statistically dominated by  a sum of independent Bernoulli
  variables with probability $1/N$ of being 1. Thus $E[{\bf Y}] \leq
  \ell/N <1$. By a standard additive Chernoff bound, the probability
  that $Y$ exceeds $\eps\ell/24$ is at most
  $\exp(-\Omega(\eps^2\ell))$. The bound on ${\bf Y}$ implies both the
  bounds in the lemma.
\end{proof}

  \begin{proof}[of Lemma~\ref{lem:payloadblocks}]
    Consider a block $x_i $ that corresponds to payload block $j$,
    that is, $x_i = B_j + e_i$. Fix $e$, $s_\ctrlset$, and $s_\pi$. The
    offset $\Delta$ is independent of these, and so we may write $x_i
    = y_i + \Delta_i$, where $y_i$ is fixed independently of
    $\Delta_i$. Since $\Delta$ is a $t'$-wise independent string with
    $t' = \Omega(\eps^2N/\log N)$ much greater than the size
    $\Lambda\log N$ of each block, the string $\Delta_i$ is uniformly
    random in $\bit{\Lambda\log N}$. Hence, so is $x_i$. By
    Proposition~\ref{prop:const-rate-AVC} we know that on input a
    random string, $\text{\sc SC-Decode}$ outputs $\perp$ with
    probability at least $1-c_1/N^2\ge 1-1/N$

    Moreover, the $t'$-wise independence of the \emph{bits} of
    $\Delta$ implies $\frac {t'}{\Lambda \log N}$-wise independence of
    the \emph{blocks} of $\Delta$. Define $t'_{blocks} = \min
    \{\frac{t'}{\Lambda \log N}, \frac{\eps \ell}{96}\}$. Note that
    $\Omega\bigl(\tfrac{\eps^2 N}{\log^2 N}\bigr) \le t'_{blocks} \le \tfrac{\eps
      \ell}{96}$. The decisions made by $\text{\sc SC-Decode}$ on
    payload blocks are $t'_{blocks}$-wise independent.  Let ${\bf Z}$
    denote the number of payload blocks that are incorrectly accepted
    as control blocks by {\sc SC-Decode}. We have $E[{\bf Z}] \le
    \frac{n'}{N} \le \eps \ell/48$ (for large enough
    $N$). 

    We can apply a concentration bound of Bellare and
    Rompel~\cite[Lemma 2.3]{BR94} using $t=t'_{blocks}$, $\mu =E[{\bf
      Z}] \le\tfrac{\eps \ell}{48}$, $A = \frac{\eps \ell}{48}$, to
    obtain the bound
    \[ 
    \Pr[{\bf Z} \geq \tfrac{\eps\ell }{24}] \leq
    8\left(\frac{t'_{blocks}\cdot \mu  +
        (t'_{blocks})^2}{(\eps\ell/48)^2}\right)^{t'_{blocks}/2} \le 
    (\log N)^{-\Omega(t'_{blocks})} \le e^{-\Omega(\eps^2 N \log \log N /
      \log^2 N)} \ . 
    \]
This bound implies the lemma statement.
\end{proof}

\begin{proof}[of Lemma~\ref{lem:controlinfo}]
  Suppose the events of Lemmas \ref{lem:controlblocks} and
  \ref{lem:payloadblocks} occur, that is, for at least $\eps \ell/4$ of
  the control blocks the recovered value $\tilde{F}_i$ is correct, at most $\eps \ell/24$
  of the control blocks are erroneously decoded, and at most
  $\eps\ell/24$ of the payload blocks are mistaken for control blocks.

  Because the blocks of the control information come with the
  (possibly incorrect) evaluation points $\tilde \alpha_i$, we are
  effectively given a codeword in the Reed-Solomon code defined for
  the related point set $\{\tilde \alpha_i\}$.  Now, the degree of the
  polynomial used for the original RS encoding is
  $d^*=|\omega|/\log(N)-1 < 3 \eps^2 N/\log N = \eps \ell /8$.  Of
  the pairs $(\tilde \alpha_i,\tilde a_i)$ decoded by {\sc SC-Decode},
  we know at least $\frac{\eps \ell}{4}$ are correct (these pairs will
  be distinct), and at most $2 \cdot \frac{\eps \ell}{24}$ are
  incorrect (some of these pairs may occur more than once, or even
  collide with one of the correct). If we eliminate any duplicate
  pairs and then run the decoding algorithm from
  Proposition~\ref{prop:RS}, the control information $\omega$ will be
  correctly recovered as long as the number of correct symbols exceeds
  the number of wrong symbols by at least $d^*+1$. This requirement is
  met if $\frac{\eps \ell }{4} - 2\times \frac{\eps\ell}{24} \ge
  d^*+1$. This is indeed the case since $d^* <
  \eps\ell/8$.

  Taking a union bound over the events of Lemmas
  \ref{lem:controlblocks} and \ref{lem:payloadblocks}, we get that the
  probability that the control information is correctly decoded is at
  least $1-\exp(-\Omega(\eps^2 N / \log^2 N))$, as desired.
\end{proof}

\subsection{Completing the Proof of Main Theorem~\ref{thm:efficient-oc-code}}
\label{sec:pf-of-main-thm}
\begin{proof}[of Theorem~\ref{thm:efficient-oc-code}]
  We will first prove that the decoding of the payload
  codeword succeeds assuming the correct control information
  $\omega=(s_\pi,s_\Delta,s_\ctrlset)$ is handed directly to the decoder,
  i.e., in the ``shared randomness'' setting. We will then account for the fact that we must
  condition on the correct recovery of the control information
  $\omega$ by the first stage of the decoder.

  Fix a message $m$, error vector $e$, and sampler seed $s_\ctrlset$, and let
  $e_Q$ be the restriction of $e$ to the payload codeword, \emph{i.e.},
  blocks not in $\ctrlset$. The relative weight of $e_Q$ is at most
  $\frac{pN}{N'} = p\frac{N'+\ell \Lambda \log N}{N'} = p (1 +
  24 \eps\Lambda \frac{N}{N'})\leq p(1+25\Lambda\eps)$ (for
  sufficiently small $\eps$). 

  Now since $s_\pi$ is selected independently from $\ctrlset$, 
  % and since the control information is assumed to be correct with
  % probability 1,
the
  permutation $\pi$ is independent of the payload error $e_Q$.
  Consider the string $\tilde P $ that is input the the $\REC$
  decoder. We can write $\tilde P
  = \tilde\pi(\tilde Q \oplus \tilde \Delta) = \pi(Q\oplus e_Q \oplus \Delta)
  $. Because a permutation of the bit positions is a linear
  permutation of $\mathbb{Z}_2^{N'}$, we get $\tilde P 
  = \pi(Q+\Delta) \oplus \pi(e_Q) = P\oplus \pi(e_Q)$. 

  Thus the input to $\REC$ is corrupted by a fraction of at most
  $p(1+25\Lambda\eps)$ errors which are $t$-wise independent, in the
  sense of Proposition~\ref{prop:concat-codes} \cite{smith07}. With probability at least $1-e^{-\Omega(\eps^2 t)} =
  1-e^{-\Omega(\eps^4 N/\log N)}$, the message $m$ is correctly
  recovered by $\text{\sc Decode}$.

In the actual decoding, the control information is not handed
  directly to the decoder. Let $\tilde{\omega}$ be the candidate control information recovered by the decoder (in Step 8 of the algorithm). 
The above suite of lemmas (Lemmas \ref{lem:sampling},
  \ref{lem:controlblocks}, \ref{lem:payloadblocks}, and
  \ref{lem:controlinfo}) show that the control information is
  correctly recovered, i.e., $\tilde{\omega}=\omega$, with probability at least 
 $\exp(-\Omega(\eps^2 N / \log^2
  N)))$.

The overall probability of success is given by 
\[ \Pr_\omega [ \mbox{payload decoding succeeds with control information  $\tilde{\omega}$} ] \]
which is at least 
\begin{align*}
\Pr_\omega [ \tilde{\omega} = \omega ~ \wedge ~ \mbox{payload decoding
  succeeds with control information  $\tilde{\omega}$} ] \\
= \Pr_\omega [ \tilde{\omega} = \omega ~ \wedge ~ \mbox{payload decoding succeeds with control information  $\omega$} ] \\
\ge  1 - \Pr_\omega [ \tilde{\omega} \neq \omega ] 
 -  \Pr_\omega [ \mbox{payload decoding succeeds given $\omega$} ] \\
\ge 1 - \exp(-\Omega(\eps^2 N / \log^2 N)) - \exp(-\Omega(\eps^4 N/\log N)) \ .
\end{align*}
 Because $\eps$ is a constant relative to $\log N$, it is the former
  probability that dominates. This completes the analysis of the
  decoder and the proof of Theorem~\ref{thm:efficient-oc-code}.
\end{proof}

\fi %MATCHES NOPROOFS

\section{Capacity-achieving codes for time-bounded channels}
\label{sec:time-bounded}
In this section, we outline a Monte Carlo algorithm that, for any
desired error fraction $p \in (0,1/2)$, produces a code of rate close to $1-H(p)$ which can be
efficiently {\em list decoded} from errors caused by an arbitrary
randomized polynomial-time channel that corrupts at most a fraction
$p$ of symbols with high probability. Recall that for $p > 1/4$,
resorting to list decoding is necessary even for very simple (constant
space) channels.

We will use the same high level approach from our construction for the 
additive errors case, with some components changed. The main
difference is that the codeword will be pseudorandom to bounded
distinguishers, allowing us to ``reduce'' to the case of oblivious
errors. (In fact, we will show that the errors are \emph{indistinguishable
from} oblivious errors, which will turn out to suffice). We will make
repeated use of the following definition:

\begin{defn}[Indistinguishability]
  For a given (possibly randomized) Boolean function $\A$ on some
  domain $D$ and two random variables $X,Y$ taking values in $D$, we
  write $X \close{\A}{\eta} Y$ if
  \[ |\Pr(\A(X)=1)-\Pr(\A(Y)=1)| \leq \eta \ . \]

We write $X \close{\text{time }\T}{\eta} Y$ to
indicate that $X\close{\A}{\eta}Y$ for all
circuits of size at most $\T$. 
\end{defn}

\begin{defn}[Pseudorandom generators]
A map $G : \{0,1\}^s \rightarrow \{0,1\}^n$ is said to be $(\T,\eps)$-pseudorandom if 
\[ G(U_s) \close{\text{time } \T}{\eps} U_b \]
where $U_m$ denotes the uniform distribution on $\{0,1\}^m$. 
\end{defn}

\vnote{Added defn of $(\T,\eps)$-pseudorandom}

We begin the section with an overview of the code construction. We
then develop several key technical results: a construction of small
``pseudorandom codes'', a useful intermediate result, the Hiding
Lemma, and finally we show how to use these to analyze the decoder's
performance.

\subsection{Code construction and ingredients}
\label{sec:code-const-time-bounded}

\mypar{Parameters} Input parameters of the construction: $N,p,\spcbound,\eps$, where
\begin{enumerate}%[(i)]
\itemsep=0ex
\item $N$ is
  the block length of the final code, 
\item $pN$ is the bound on the number of errors
  introduced ($0 < p  < 1/2$) by the channel w.h.p.
\item $\T_0$ is a bound on the circuit size (think ``running time'') of
  the channel. We will eventually set $\T_0=N^c$ for a constant $c>1$,
  though most results  hold for any $\T_0>N$.
\item $\eps$ is a measure of how far the rate is from the optimal
  bound of $1-H(p)$ (that is, the rate must be at least
  $1-H(p)-\eps$). We will assume $0 < 2\eps < 1/2-p$. 
\end{enumerate}
%\end{inparaenum}

\mypar{The seeds/control information}
The control information $\omega$  consists of three randomly chosen strings
$s_\pi,s_\ctrlset,s_\Gamma$ 
where $s_\pi,s_\ctrlset$ are as in the additive errors case. We take the
lengths of $s_\pi,s_\ctrlset$ be $\zeta N$ where $\zeta=\zeta(p,\eps)$ will be
chosen small enough compared to $\eps$.  \anote{Can we say a specific bound, maybe $\eps^{ 5}$?}

The third string $s_\Gamma \in \{0,1\}^{\gamma N}$ is the seed of a
pseudorandom generator $\GEN$ that outputs $N$ bits and  fools
circuits of size (roughly) $\T_0$. (A shorter seed would suffice here,
but we make all seeds the same length for simplicity -- see below.)
% Such a generator exists by a counting argument,
% and we will find one with high probability by selecting a uniformly 
% random function from $O(\log \T_0)$ to $N$ bits. 
%
The offset $\Gamma \gets
\GEN(s_\Gamma)$ will be used to fool the polynomial-time channel. We
do not need to add the $t'$-wise independent offset $\Delta$ as we did
in the additive errors case.

\mypar{Encoding the message}
The payload codeword encoding the message $m$ will be
$\pi^{-1}(\mathsf{REC}(m)) \ \oplus \Gamma$, which is the same as the
encoding for the additive channel, with the offset $\Gamma$ added to
break dependencies in the time-bounded channel instead of the
$t'$-wise independent offset $\Delta$. 

The code $\REC$ is the same as the code for the case of additive
errors.  We will denote by $\beta_{\REC}$ the maximum, over fixed
error patterns $e$, of the probability, over permutations $\pi$, that
$\REC$ does not correctly decode the pattern $\pi(e)$. As noted in 
Section~\ref{sec:ingre-brief}, $\beta_\REC\leq 2^{-\Omega(\eps^2
  \zeta N/\log N)}$.
%
% \anote{Maybe defer this to proof of payload decoding lemma?}  
We will need the following
additional property of $\REC$: there is a circuit of size $O_{\eps}(N)$ that
takes as input an error pattern $e$, permutation $\pi$ and set of
control positions $V$, and checks whether or not $\REC$ will decode
the error pattern $\pi(e)$ (restricted to positions outside of $V$)
correctly. (This circuit works by counting the number of blocks of the
concatenated code which the inner code decodes incorrectly.)

For the offset $\Gamma$, we need a pseudorandom generator $\polyprg$
that is computable in time $\poly(N)$ and secure against circuits of
size $T$ with polynomially small error, with seed length at most
$\gamma N$. Such generators can be constructed in a Monte Carlo
fashion (a random function from $c\log N$ to $N$ bits will do, for a
large enough constant $c$):

\begin{prop}[Folklore (also follows from Proposition~\ref{prop:inner-time-bounded})]
  For every constant $\zeta>0$ and polynomial $\T(N)\geq N$, for
  sufficiently large $N$ there exists a poly-time Monte Carlo construction of a
  polynomial-time computable function $\polyprg$ from $\zeta N$ to $N$
  bits that is $(\T,\frac 1 \T)$ pseudorandom with probability at
  least $1-1/\T$.
\end{prop}

This construction can be made explicit assuming either that one-way
functions exist~\cite{yao82,HILL}, or that $\mathsf{E} \not\subseteq
\mathsf{SIZE}(2^{\eps_0 n})$ for some absolute constant $\eps_0 > 0$
\cite{IW97} (where $\mathsf{E} = \mathrm{TIME}(2^{O(n)})$ and
$\mathsf{SIZE}(2^{\eps_0 n})$ denotes the class of languages that have
size $O(2^{\eps_0 n})$ circuits). For \emph{space-bounded} (as opposed to time-bounded)
distinguishers,
there is even an explicit construction that makes no
assumptions~\cite{nisan}. However, as noted by one reviewer, we
require a Monte Carlo algorithm to construct pseudorandom codes (as required by Proposition \ref{prop:inner-time-bounded} below), even
for space-bounded channels. Therefore, we use a Monte Carlo construction of
$\polyprg$ and get a single statement covering time- and space-bounded
channels.

\mypar{Encoding the seeds}
The control information (consisting of the seeds $s_\pi,s_\ctrlset,s_\Gamma$)
will be encoded by a similar structure to the solution for the
additive channel: a Reed-Solomon code of rate $R^{\RS} =
R^{\RS}(p,\eps)$ concatenated with an inner stochastic code. But the
stochastic code $\SC$ (of Proposition \ref{prop:const-rate-AVC}) will
now be replaced by a {\em pseudorandom code} $\PRC$ which satisfies
two requirements: first, it has
good list decoding properties and, second,  the (stochastic) encoding of
any given message  is indistinguishable from a random
string by a randomized time-bounded channel.
The construction of the necessary stochastic code is guaranteed by the following
lemma.
% \begin{defn}[List-decodable  pseudorandom code]
% \label{def:lsc} 
% A binary stochastic code with encoding map $E: \{0,1\}^k \times
% \{0,1\}^s \rightarrow \{0,1\}^b$  is
% said to be a {\em $(\delta,L)$-list decodable
%   $(\T,\beta)$-pseudorandom code} if the following properties hold:
% \begin{enumerate}
% \item $E$ is $(\delta,L)$-list decodable, i.e., for every $y\in
%   \bit{b}$, there are at most $L$ pairs $(m,r)$ such that $E(m,r)$ is
%   within Hamming distance $\delta b$ of $y$.
% \item For every $m \in \{0,1\}^k$:
% \[ E(m,U_s) \close{\text{time }\T}{\beta} U_b \ . \]
%  That is, a random encoding of $m$ is indistinguishable from
%  uniform with in $\beta$ by size-$\T$ randomized
%   circuits $\C$, even if they depend on $m$.
% \end{enumerate}
% The \emph{rate} of such a code is $k/b$, and its {\em seed
%   length} is $s$.\closedef
% \end{defn}

\begin{prop}[Inner control codes exist]
\label{prop:inner-time-bounded}
For some fixed positive integer $\Lambda_0$ the following holds.  For
all $\delta$, $0 < \delta < 1/2$ 
and polynomials $\T=\T(N)\geq N$, 
for sufficiently large $N$ there exist $R = R(\delta) \ge
(1/2-\delta)^{\Omega(1)} >0$ and a positive integer $L=L(\delta) \le
1/(1/2-\delta)^{O(1)}$ such that there exists a $\poly(\T)$ time
randomized Monte Carlo algorithm that
outputs a  stochastic code 
% with encoding $E: \{0,1\}^k \times \{0,1\}^s \rightarrow \{0,1\}^{b}$
 with block length
$b = \Lambda_0
\log(\T)$ and rate $k/b\geq R$ 
%and $s = 10 \log(\T)$, 
that is 
%$(\delta,L)$-list-decodable and $(\T,\frac1{\T})$-pseudorandom
%with probability at least $1-2^{-b}$.
\begin{itemize}
\item  $(\delta,L)$-list decodable: for every $y\in
  \bit{b}$, there are at most $L$ pairs $(m,r)$ such that $E(m,r)$
  is within Hamming distance $\delta b$ of $y$;
\item $(\T,1/\T)$-pseudorandom: For every $m \in \{0,1\}^k$, we have
  \( E(m,U_s) \close{\text{time }\T}{1/\T} U_b  \). 
% That is, a
  % random encoding of $m$ is indistinguishable from uniform with in
  % $\beta$ by size-$\T$ randomized circuits $\C$, even if they depend
  % on $m$.
\end{itemize}

Further, there exists a deterministic decoding procedure running in
time $\poly(\T)$ that, given a string $y\in \bit{b}$, recovers the
complete list of at most $L$ pairs $(m,r)$ whose encodings $E(m,r)$
are within Hamming distance at most $\delta b$ from $y$.
\end{prop}

Proposition~\ref{prop:inner-time-bounded} is proved in Section~\ref{sec:inner-code-time-bounded}.
An interesting direction for future work is the design of \emph{explicit}
pseudorandom stochastic codes along the lines of
Proposition~\ref{prop:inner-time-bounded}. See Section~\ref{sec:openq}.

\mypar{Full Code}
To construct the final code, we will apply the Monte Carlo
constructions of the previous section with time $\T_2 = \T_0 +
O(N\max\{\log N, 2^{\poly(1/\eps)}\})$ (the exact value of $T_2$
will be clear from the analysis). For constant $\eps$, it suffices to
take $\T_2 = 2T_0$ when $\T_0$is a large enough polynomial in
$N$. This means the control blocks have length $\Lambda_0\log T_2
=\Theta(\log T_0)$. The control blocks occupy an $\eps$ fraction of
the whole codeword, which means the number of real control blocks is
$n_\ctrl =\Theta(\frac {\eps N}{\log T_0})$. 

As in the additive errors
case, the control blocks will be interspersed with the payload blocks
at locations specified by the sampler's output on $s_\ctrlset$.

\mypar{Rate of the code} 
The code encoding the control
information is of some small constant rate $R_\ctrl(p,\eps)$, but the
control information consists only of $O(\zeta N)$ bits. Given $\eps$,
we can select $\zeta$ small enough so that the control portion of the
codeword only adds $\eps N/2$ bits to the overall encoding. The rate of
$\mathsf{REC}$ is at least $(1-\eps/10)(1-H(p)-\eps/10) \ge
1-H(p)-\eps/5$. So the rate of the overall stochastic code is at least
$1-H(p)-\eps$ as desired.

\subsection{List decoding algorithm for full encoding}
\label{sec:dec-algo-time-bounded}

The decoding algorithm for the full encoding will be similar to the
additive case with the principal difference being that the inner
stochastic codes will be decoded using the procedure guaranteed in
Proposition~\ref{prop:inner-time-bounded}. For each block, we obtain a
list of $L$ possible pairs of the form $(\alpha_i, a_i)$.  This set of
(at most $N L$) pairs is the fed into the polynomial time Reed-Solomon
list decoding algorithm (guaranteed by
Proposition~\ref{prop:RS-list}), which returns a list of ${\rm
  poly}(1/\eps)$ values for the control information. This comprises
the {\em first phase} of the decoder.

Once a list of control vectors is recovered, the {\em second phase} of the
decoder will run the decoding algorithm for $\REC$ for each 
of these choices of the control information and recover a list of
possible messages.

The steps to decode each of inner stochastic codes takes time
$\poly(\T)$ and decoding the Reed-Solomon code as well as $\REC$ takes
time polynomial in $N$. So the overall run time is polynomial in $N$
and $\T$.

\begin{theorem}
\label{thm:decoding-time-bounded}
Let $\channel_{\T_0}$ be an arbitrary randomized time-${\T_0}$ channel
on $N$ input bits that is $pN$-bounded.
Consider the code construction described in Section
\ref{sec:code-const-time-bounded}.
%  using component codes
% $\mathsf{REC}$, a Reed-Solomon code of small enough rate
% $R^{\RS}$, the code $\PRC$ that is
% $(p+\eps,L)$-list decodable and $({\T}_2,1/{\T}_2)$-pseudorandom
% (which happens with $1-\poly(1/{\T_0})$ probability) and the generator
% $\polyprg$ that is $({\T}_2,1/{\T}_2)$-pseudorandom.

The resulting code has rate at least $1-H(p)-\eps$. For
every message $m$, with high probability over the choice of control
information $s_\pi,s_\ctrlset,s_\Gamma$, the coins of $\PRC$ and the
coins of
$\channel_{\T_0}$, the list output by the  decoding
algorithm has size $\poly(1/\eps)$ and includes the message $m$ with probability at least
\[ 1 - O(N/{\T_0})\ . \]
The running time of the decoding algorithm is polynomial in $N$ and
${\T_0}$.
\end{theorem}

The novelty compared to the additive errors case is in the analysis of
the decoder, which is more subtle since we have to deal with a much
more powerful channel. The remainder of this section deals with this
analysis, which will establish the validity of
Theorem~\ref{thm:decoding-time-bounded}.

\subsection{Monte Carlo Constructions of Pseudorandom Codes}
\label{sec:inner-code-time-bounded}

\begin{proof}[of Proposition~\ref{prop:inner-time-bounded}]
The codes we design have a specific structure: 
  A binary stochastic code with encoding map $E$ where $E: \{0,1\}^k
  \times \{0,1\}^s \rightarrow \{0,1\}^b$ is said to be {\em
    decomposable} if there exist functions $E_1 : \{0,1\}^k
  \rightarrow \{0,1\}^b$ and $E_2 : \{0,1\}^s \rightarrow \{0,1\}^b$
  such that $E(x,y) = E_1(x) \oplus E_2(y)$ for every $x,y$. We say
  that such a encoding decomposes as $E=[E_1,E_2]$.

  The existence will be shown by a probabilistic construction with a
  decomposable encoding $E(m,r) = C(m) \oplus \lsg(r)$ where $C$ will
  be (the encoding map of) a linear list-decodable code, and $\lsg$
  will be a generator that fools size $\T$ circuits,
  obtained by picking $\lsg(r) \in \{0,1\}^b$ independently and
  uniformly at random for each seed $r$. Here $b = \Lambda_0 \log(\T)$ for
  a large enough absolute constant $\Lambda_0$ as in the statement of
  the Proposition. Note that the construction time is
  $2^{O(b)}=\poly(\T)$.

  \medskip {\sc List-decoding property.} We adapt the proof that a
  truly random set is list-decodable. Let $C \subseteq \{0,1\}^b$ be a
  linear $(\delta,L_C=L_C(\delta))$-list decodable code; such codes
  exist for rates less than $1-H(\delta)$~\cite{GHSZ}, and can be
  constructed explicitly with positive rate $R(\delta) \ge
  (1/2-\delta)^{\Omega(1)} > 0$ for any constant $\delta < 1/2$ with a
  list size $L_C \le 1/(1/2-\delta)^{O(1)}$~\cite{GS2000}. We will
  show that the composed code $E$ has constant list-size with high
  probability \emph{over the choice of $\lsg$} as long as the rate of
  the combined code is strictly less than $1-H(\delta)$.

  Fix a ball $B'$ of radius $\delta b$ in $\bit{b}$, and let $X$ denote
  the size of the intersection of the image of $E$ with $B'$.  We can
  view the image of $E$ as a union of $2^s$ sub-codes $C_r$, where
  $C_r$ is the translated code $C \oplus \lsg(r)$ (for
  $r\in \bit{s}$). Each sub-code $C_r$ is $(\delta,L_C)$-list-decodable
  since it is a translation of $C$.
We can then write $X=\sum_{r\in\bit{s}} X_r$, where $X_r$ is the size
of $C_r \cap B'$. The $X_r$ are independent integer-valued random
variables with range $[0,L_C]$ and expectation $\expec[X_r]=|C|\cdot
|B'| / 2^b \le 2^{-b(1-H(\delta)-R_C)}$ where $R_C$ denotes the rate of
  $C$.

If we set $s=10\log T_0$, we get:
\[ \expec[X]= 2^{10\log(\T)} 2^{-b(1-H(\delta)-R_C)} = 2^{-b(1-H(\delta)-R_C-10/\Lambda_0)} \ . \]

Suppose $R_C+10/\Lambda_0=1-H(\delta)-\alpha_0$, so that $\expec[X]
=2^{-\alpha_0 b}$.  Let $t$ be the ratio $L/\expec[X]$, where
$L=L(\delta)$ is the desired list-decoding bound for the composed code
$E$. We will set $L = 3 L_C/\alpha_0$.  By the multiplicative Chernoff
bound for bounded random variables, the probability (over the choice
of $\lsg$) that $X>L$ is at most
$\left(\frac{t}{e}\right)^{t\expec[X]/L_C}.$ Simplifying, we get
$\Pr[X>L] \leq (\frac L e)^2 2^{-3b} \le 2^{-2b}$.

Taking a union bound over all $2^b$ possible balls $B'$, we get that
with probability at least $1-2^{-2b}$, the random choice of $\lsg$
satisfies the property that the decomposable stochastic code with
encoding map $E= C\xor \lsg$ is $(\delta, L)$-list-decodable.

\medskip
{\sc Pseudorandomness.} The proof of pseudorandomness  is
standard, but we include it here for completeness.
It suffices to prove the pseudorandomness property against all
deterministic circuits of size $\T$, since a
randomized circuit is just a distribution over
deterministic ones.

Fix an arbitrary codeword $C(m)$.  Consider the (multi)set $X_m = \{C(m)
\oplus \lsg(r)\}$ as $r$ varies over $\{0,1\}^s$. Each element of this
set is chosen uniformly and independently at random from $\{0,1\}^b$.
Fix a circuit $B$ of size $\T$.  By a standard
Chernoff bound, the probability, over the choice of $X_m$, that
$\Pr_{x\in X_m} [ B(x) =1]$ deviates from the probability $\Pr
[B(U_b)=1]$ that $B$ accepts a uniformly random string by more than
$\zeta$ in absolute value, is at most $\exp(-\Omega(\zeta^2
|X_m|))$. For $\zeta = 1/\T$ and $|X_m| = 2^s \ge \T^{10}$, this
probability is at most $\exp(-\Omega(T_0^8))$.

The number of circuits of size $\T$  is  $\exp(O(\T\log(\T)))$.  By a union
bound over all these branching programs, we conclude that except with
probability at most $\exp(-\poly(\T))$ over the choice of $\lsg$,
the following holds for {\em every} size-$\T$ circuit $B$:
\( | \Pr_{x\in X_m} [ B(x) =1] - \Pr [B(U_b)=1] | \le 1/\T \ . \)
Since $m$ was arbitrary, we have proved that the constructed
stochastic code is $(\T,1/\T)$-pseudorandom with probability at
least $1-2^{-2b}$.

\medskip {\sc Decoding.} Finally, it remains to justify the
  claim about the decoding procedure. Given a string $y \in
  \{0,1\}^b$, the decoding algorithm will go over all $(m,r)\in
  \{0,1\}^k \times \{0,1\}^s$ by brute force, and check for each
  whether $\hamdist(E(m,r),y) \le \delta b$. By the list-decoding
  property, there will be at most $L$ such pairs $(m,r)$.  The
  decoding complexity is $2^{O(k+s)} = 2^{O(b)}$.
\end{proof}

\subsection{Analyzing Decoding: Main Steps}

It will be convenient to explicitly name the different sources of
error in the construction. The code ingredients are selected so that
each of these terms is at most $1/\T_0$.

\begin{center}
  \begin{tabular}{c|p{0.7\textwidth}}
    \hline
    $\beta_\Gamma(\T)$ & distinguishability of long generator (from
    uniform) by circuits of size $\T$\\
    \hline
    $\beta_\PRC(\T)$ & distinguishability of $\PRC$ outputs by circuits of size $\T$
    \\
    \hline
    $\beta_\ctrlset$ & max. probability (over choice of sampler set $V$)
    that a fixed error pattern will not be well-distributed among
    control blocks \\
    \hline
    $\beta_{\pi}$ & max. probability that any fixed error pattern will
    cause a decoding error for code $\REC$ after $\pi$\\
    \hline
  \end{tabular}
\end{center}

Our analysis requires two main claims.
\begin{lemma}[Few Control Candidates]
\label{lem:controllist}
  The decoder recovers a list of $L' \le \mathrm{poly}(1/\eps)$
 candidate values of the control information. The list includes the correct value $\omega=(s_\pi,s_\ctrlset,s_\Gamma)$ of
  the control information used at the encoder with probability at least $1-\beta_{control}$, where
  $$ \beta_{control} \leq  \beta_\ctrlset+ \beta_\Gamma(\T_2) + N\cdot
  \beta_\PRC(\T_2) \leq (N+3)/\T_0$$ when
  $T_2 = T_0 + O(N\log  N)$.
\end{lemma}

\begin{lemma}[Payload decoding succeeds]\label{lem:payloaddecode}
  Given the \emph{correct} control information $(\pi,\ctrlset,\Gamma)$, the
  decoder recovers the correct message with high probability. Specifically, the
  probability of successful decoding is at
  least $1-\beta_{payload}$ where $$\beta_{payload} \leq  \beta_{\pi} + \beta_\Gamma(\T_2) + N 
  \cdot\beta_\PRC(\T_2)) \leq (N+3)/\T_0$$ and
  $T_2 = T_0 + O(N2^{\poly(1/\eps)})\, .$
\end{lemma}

Combining these two lemmas, which we prove in the next two sections,
we get that except with probability at most $\beta_{Control} +
\beta_{payload} \leq \beta_\ctrlset +
\beta_{\pi} + 2\beta_\Gamma(\T_2) + 2N\cdot \beta_\PRC(\T_2)\leq 2(N+3)/\T_0$, the decoder
recovers a list of at most $L' \le \mathrm{poly}(1/\eps)$ potential
messages, one of which is the correct original message. 
This
establishes Theorem \ref{thm:decoding-time-bounded}.

\subsection{The Hiding Lemma}

%Let $\ell =\Theta(\eps N/b_\ctrl)$ be the number of control blocks.
Given a message $m$, and pseudorandom outputs $\pi,\ctrlset,
\Gamma$ based on the seeds $s_\pi,s_\ctrlset,s_\Gamma$, let 
\[ \Enc(m;\pi,\ctrlset,\Gamma,r_1,...,r_{n_\ctrl}) \]
denote the output of the encoding algorithm when the $r_i$'s are used
as the random bits for the $\LSC$ encoding. Let $\Enc(m;\pi,\ctrlset,\cdot)$
be a random encoding of the message $m$ using a given $\pi,\ctrlset$ and
selecting all other inputs at random.

% \anote{DOn't think we need this definition anymore. Yeah!}
% \begin{defn}[Conditional Indistinguishability]
%   For random variables $X,Y,Z$ with $X,Y$ defined on $\bit{N}$, and
%   $\eta \ge 0$, we say that $X$ and $Y$ are \emph{time $\T$
%     indistinguishable given $Z$ with advantage $\eta$} if for all
%   values $z$ of $Z$, and for all circuits
%   $\A_z$ (that could depend nonuniformly on $z$) with $N$ inputs and
%   size at most $\T$, we have $X\close{\A_z}{\eta}Y$, where $X$ and $Y$ are
%   conditioned on $Z=z$. \closedef
% \end{defn}

%
\begin{lemma}[Hiding Lemma]
\label{lem:hiding}
% For all messages $m$, permutations $\pi$ and sampled sets $\ctrlset$, and for
% all space-$s$ online distinguishers ${\A}_{m,\pi,\ctrlset}$, possibly
% depending nonuniformly on $m,\pi,\ctrlset$, we have $$\Enc(m;\pi,\ctrlset,\cdot))
% \close{\A_{m,\pi,\ctrlset}}{\eta} U_n\,,$$ where $U_n$ is uniform on
% $\bit{n}$.
For all messages $m$, sampler sets $\ctrlset$ and permutations
$\pi$,  the random variable
 $\Enc(m;\pi,\ctrlset,\cdot)$ is pseudorandom, namely:
$$\Enc(m;\pi,\ctrlset,\cdot)\close{\text{time }\T}{\beta} U_N\,,$$
% the random variables 
% $\Enc(m;\pi,\ctrlset,\cdot)$ and $U_N$ (a fresh draw from
% uniform distribution on $\bit{N}$) are time $\T$
% indistinguishable given $(m,\pi,\ctrlset)$ with advantage
% $\beta_{Hide}$, 
where
$\beta \leq \beta_{Hide}(\T) \defeq  \leq N\cdot \beta_\PRC(\T') + \beta_\Gamma(\T')$ and $\T' =
\T + N$. 
\end{lemma}

For our purposes, the most important consequence of the Hiding Lemma 
that even with the knowledge of $m, \pi$ and $\ctrlset$, the distribution of
errors inflicted by a space-bounded channel on a codeword of our code and on a
uniformly random string are indistinguishable by time-bounded
tests. 

In other words, the pseudorandomness of the codewords allows us to
reduce the analysis of time-bounded errors to the additive case, as
long as the events whose probability we seek to bound can be tested
for by small circuits.
 We
encapsulate this idea in the ``Oblivious Errors Corollary'',
below. The proof of the Hiding Lemma follows that of the corollary.

\begin{defn}[Error distribution]
  Given a randomized channel $\channel$ on $N$ bits and
  a random variable $D$ on $\bit{N}$, let $\E_\channel(D)= D\xor \channel(D)$ denote the error
  introduced by $\channel$ when $D$ is sent through the channel.
\end{defn}

% (For these later applications, it suffices to fool deterministic
% branching programs, so we state the claim for deterministic machines
% to make the task of proving it potentially cleaner.)

\begin{cor}[Errors are Near-Oblivious]
\label{cor:error-dist-fools-L} For every $m,\pi,\ctrlset$ and  randomized time-$\T$ channel $\channel$.
The error distributions for a real codeword and a random string are
indistinguishable. That is, for every time bound $\T'$:
$$\E_{\channel}(U_N) \close{\text{time }\T'}{\beta}
\E_{\channel}(\Enc(m;\pi,\ctrlset,\cdot))$$ 
where $\beta \leq \beta_{Hide}(\T + \T' + N)$ and $\beta_{Hide}(\cdot)$
is the bound from the Hiding Lemma (\ref{lem:hiding}). Note that the distinguishing circuits here may depend on $m,\pi,\ctrlset$.
\end{cor}

\begin{proof}[of Corollary]
  One can compose a distinguisher for the two error random variables
  with the channel $\channel$ to get a distinguisher for the
  original distributions of the Hiding Lemma. This composition
  requires the addition of a layer of XOR gates (thus adding $N$ to
  the size of the circuit).
\end{proof}

\begin{proof}[of  the Hiding Lemma (Lemma~\ref{lem:hiding})]
  The proof proceeds by a standard hybrid argument. Fix
  $m,\pi,\ctrlset$, and recall that $n_\ctrl=|\ctrlset|<N$ is the
  number of control blocks. Let $D_0$ be the random variable
  $\Enc(m;\pi,\ctrlset,\cdot)$, and $D_{n_\ctrl+1}$ be the uniform
  distribution over $\bit{N}$ .  We define intermediate random
  variables $D_1,D_2,...,D_{n_\ctrl}$: In $D_i$, the first $i$ control
  blocks from $D_0$ are replaced by random strings.  For a given
  time-$\T$ circuit $\A$, let $p_i= \Pr[\A(D_i)=1]$.

  Note that for $i\in \{1,...,n_\ctrl\}$, $D_{i-1}$ and $D_{i}$ are
  distributed identically in all blocks except the $i$th-control
  block.  The pseudorandomness of $\PRC$ implies that, conditioned on
  any particular fixed value of the positions outside the $i$-th
  control block, the distributions $D_{i-1}$ and $D_i$ are
  indistinguishable up to error $\beta_{\PRC}(\T)\leq
  \beta_{\PRC}(\T')$ by circuits of size $\T$. Averaging over the
  possible values of the other blocks, we get $|p_i - p_{i-1}|\leq
  \beta_\PRC(\T')$.

  To compare $D_{n_\ctrl}$ and $D_{n_\ctrl +1} = U_N$, note that the
  two distributions would be identical if the offset $\Gamma$ were
  replaced by a truly random string. Since the control blocks are now
  random (and hence carry no information about $s_\Gamma$), we use a
  distinguisher for $D_{n_\ctrl}$ and $D_{n_\ctrl +1}$ (of size $\T$)
  to get a distinguisher between $\Gamma$ and $U_{N'}$ (of size $\T'$)
  by XORing the challenge string with the $\REC$ encoding of the
  message, permuted according to $\pi$. Because $m,\pi,\ctrlset$ can
  be fixed, the encoding can be hardwired into the new distinguisher,
  leading to a size increase of only $N'<N$ XOR gates. Thus,
  $|p_{n_\ctrl+1}-p_{n_\ctrl}|\leq \beta_\Gamma(\T')$ where $\T' = \T+N$.

  Combining these bounds, we get $|p_{n_\ctrl+1} - p_0| \leq N
  \beta_\PRC(\T') + \beta_\Gamma(\T')$, as desired.
\end{proof}

\subsection{Control Candidates Analysis}

Armed with the Hiding Lemma, we can  show that the decoder can recover a small list of
candidate control strings, one of which is correct with high
probability (Lemma \ref{lem:controllist}).

There were four main lemmas in the analysis of additive errors.  The
first  (Lemma~\ref{lem:sampling}) stated that the sampler set
$\ctrlset$  is
\emph{good}  error pattern $e$ for $\ctrlset$ with high probability. 
A version of this lemma
holds also for time-bounded errors.  
Recall that 
$\ctrlset$ is good for $e$
(Definition~\ref{def:goodsampler}) if at least a fraction $\eps$ of
the $n_\ctrl$ control blocks have an error rate (fraction of flipped
bits) bounded above by $p+\eps$.

\begin{lemma}[Good Samplers: time-bounded analogue to Lemma~\ref{lem:sampling}]
  \label{lem:sampling2corrected} For every $pN$-bounded time-$\T_0$
  channel with $\T_0 > N$ and for every message $m$ and permutation
  seed $s_\pi$,  the set $\ctrlset$ is good
  for the error pattern $e$ introduced by the channel 
  with probability at least $1- (\beta_\ctrlset + 2\beta_\Gamma(\T_2) + 2N\cdot \beta_\PRC(\T_2))\geq 1-2(N+3)/\T_0 $ over the choice of
  sampler seed $s_\ctrlset$, the seed $s_\Gamma$ for the pseudorandom offset, the coins of the
  control encoding and the coins of the channel. Here $\T_2 = \T_0 +
  N\log N$. 
\end{lemma}

\begin{proof}
  The crucial observation here is that Oblivious Errors Lemma implies
  that the error pattern introduced by a bounded channel is
  almost independent of $V$ from the point of view of any time-$\T$
  test. That is, by a simple averaging argument, the Oblivious Errors Lemma implies that for every
  distribution on $m,\pi,\ctrlset$, we have
  \begin{equation}
  (m,\pi,\ctrlset,\E_{\channel}(U_N)) \close{\text{time }\T'}{\beta}
  (m,\pi,\ctrlset,\E_{\channel}(\Enc(m;\pi,\ctrlset,\cdot)))\,.\label{eq:sampler}
\end{equation}

  The properties of the sampler imply that the probability of getting
  a good (control set, error) pair in the left hand distribution is at
  least $1-\beta_\ctrlset$. 
  
  Thus, all we really have to do is show that ``goodness'' of
  the control set $\ctrlset$ can be tested efficiently, given $e$ and
  $\ctrlset$. Testing for goodness only involves counting the number
  of errors in each of the control blocks, and tallying the number of
  control blocks with too high a fraction of errors. This can be done
  in circuit size $O(N \log N)$ (in fact, quite a bit less but the
  optimization doesn't matter here). 

  By the Oblivious Errors lemma, the probability of a good (control
  set, error) pair on the right-hand side of \eqref{eq:sampler} is at
  least $1-\beta_\ctrlset - \beta_{Hide}(\T+O(N\log N))$, as desired.
\end{proof}

We now turn to the second lemma in the analysis of the
control information decoding (Lemma \ref{lem:controlblocks} for
additive errors), which stated that when the sampled positions
$\ctrlset$ are good for the error pattern $e$, one can correctly
recover a large number of control blocks. This is no
longer true in our setting, but we require only the following weaker
statement.

\begin{lemma}[Correct Control Blocks --- list decoding version]
\label{lem:LScontrolblocks}
    For any fixed $e$ and $\ctrlset$ such that $\ctrlset$ is good for $e$,
    the decoding algorithm for the inner codes $\PRC$ outputs a list
    of $L$ symbols containing the correct symbol $a_i$ for at least
    $\frac{\eps n_\ctrl}{2}=\Theta(\frac{\eps^2 N}{\log\T_0})$ control
    blocks.
\end{lemma}
\begin{proof}
  The list-decoding radius of the $\PRC$ code is set to be $\delta>p+\eps$, so all blocks with an error rate below $p+\eps$ produce a valid list.
\end{proof}

The third lemma from the analysis of additive errors (Lemma
\ref{lem:payloadblocks}), which previously stated that very few
payload blocks are mistaken for control blocks, requires a significant
change, because it is possible for
the time-bounded channel to inject fake control blocks into the
codeword (by changing a block to some pre-determined codeword of
$\PRC$). Therefore we can only say that the total number of candidate
control blocks is small. 
%(Note that our lower bounds show that some version of this injection attack is unavoidable.)

\begin{lemma}[Bounding mistaken control blocks]
\label{lem:LSpayloadblocks}
  For every $m,e,\omega$, the total number of candidate control
  symbols is at most $\frac{N L}{b_\ctrl} = \Theta (\frac{NL}{\log T_0})$.
 \end{lemma}
 \begin{proof}
   Since each candidate control block has $b_\ctrl = \Theta(\log \T_0)$ bits, there are
   $\frac{N}{\log \T_0}$ blocks considered by the decoder. The
   list decoding of each such block yields at most $L$ candidate
   control symbols.
 \end{proof}

 \begin{proof}[of Lemma~\ref{lem:controllist}]
   Given Lemmas \ref{lem:sampling2corrected},
   \ref{lem:LScontrolblocks} and \ref{lem:LSpayloadblocks}, we only
   need to ensure that the rate $R^{\RS}$ of the Reed-Solomon code
   used at the outer level to encode the control information is small
   enough so that list decoding is possible according to Proposition
   \ref{prop:RS-list} as long as (1) the number of data pairs $n$ is
   at most $ \frac{N L}{b_{ctrl}}$ and (2) the number of agreements
   $t$ is at least $\Theta(\frac{\eps^2 N}{\log \T_0})$. The claimed
   list decoding is possible with rate $R^{\RS}=O(\eps^4/L)$, and the
   list decoder will return at most $O(L/\eps^2)$ candidates for the
   control information. Since the list decoding radius $\delta$ of
   $\PRC$ was chosen to be $\delta = p + \eps < 1/2-\eps$, we have $L
   \le 1/\eps^{O(1)}$ by the guarantee of
   Proposition~\ref{prop:inner-time-bounded}, so the output list size
   is bounded by a polynomial in $1/\eps$. This proves
   Lemma~\ref{lem:controllist}.
 \end{proof}

\subsection{Payload Decoding Analysis}

We now show Lemma~\ref{lem:payloaddecode}: given a magical, correct copy of the control
information, the payload blocks will correctly be decoded to recover
the message $m$.
The assumption of correct control information essentially places us in
the \emph{shared randomness}
setting.

As with the analysis of the control block decoding, we use the fact
that errors are nearly-oblivious (Corollary~\ref{cor:error-dist-fools-L}) to argue that the events that
we needed to happen with high probability for successful decoding
against additive errors (where the errors were oblivious to the
codeword) will also happen with good probability with a time-bounded channel.

\begin{proof}[of Lemma~\ref{lem:payloaddecode}]
Fix the message $m$ and the choice $\ctrlset$ of the
control block locations. Recall that the probability for
an \emph{oblivious} error distribution that $\pi(e)$ causes $\REC$ to
decode $m$ incorrectly is at most $\beta_{\pi} \leq 1/\T_0$.   

Note that given a permutation $\pi$, control
set $\ctrlset$ and an
error pattern $e$, one can easily check if the payload code $\REC$
will correctly decode the  error pattern (one could run the
full decoder for $\REC$; alternatively, one can simply check that a
sufficiently large number of the inner code blocks are decoded
correctly by the inner code). There is a circuit of size $O(N
2^{\poly(1/\eps)})$ that verifies if decoding will occur correctly. 

We can thus apply the Oblivious Errors Corollary
(\ref{cor:error-dist-fools-L}) with $\T=\T_0$ and $\T' = O(N
2^{\poly(1/\eps)})$ to show the following: for every fixed
$m,\ctrlset$, the probability that a time-$\T_0$ channel $\channel$
introduces errors that induce a decoding mistake (by $\REC$) is at
most 
$$\beta_\pi + \beta_{Hide}(\T_2) = \beta_\pi + \beta_\Gamma(\T_2) +
N\beta_\PRC(\T_2)\, ,$$
  where $\T_2= \T_0 + O(N
2^{\poly(1/\eps)})$, as desired.
\end{proof}

%MATCHES OLD-CRAP

\section{Open Questions}
\label{sec:openq}

%\anote{This section rewritten heavily.}
The code constructions of the previous sections
leave several open questions:

\mypar{Uniquely decodable codes beyond the GV bound}
The codes we design for time-bounded channels are list-decodable, but
not necessarily uniquely decodable. This is inherent to the current
analysis, since even a very simple adversary may inject valid control
blocks into the codeword, potentially causing the decoder to come up
with several seemingly valid control strings. For $p\geq 1/4$, we know
the limitation is inherent to \emph{any} construction, because our
lower bound describes an attack that can be carried out by a very
simple attacker. However, for $p<1/4$, it may still be possible to
design codes that lead the decoder to a unique, correct codeword with
high probability. 
Since the initial version of this paper was published,   \cite{HavivL11} 
showed that random codes can tolerate causal errors slightly beyond
the Gilbert-Varshamov bound (regardless of the channel's complexity) in the low-error regime (i.e., they can achieve a rate better than $1-H(2p)$ for correcting a fraction $p$ of online errors, for small $p$). 
Those codes are neither explicit nor decodable in polynomial time, however.

\mypar{(More) Explicit Constructions for Time-Bounded Channels} 
Our design of time-bounded channels uses Monte Carlo constructions in
two places: for the pseudorandom code $\PRC$ and the generator
$\Gamma$. Constructions for the generator $\Gamma$ are in fact known
based on worst-case hardness assumptions for circuits (say, that
exponential time does not have subexponential size circuits~\cite{IW97}). 
An interesting direction for future work is the design of
\emph{explicit} pseudorandom stochastic codes along the lines of
Proposition~\ref{prop:inner-time-bounded} under such hardness
assumptions. This would make the entire code construction explicit (conditionally).

\mypar{Explicit Constructions for Online Space-Bounded Channels}
Similarly to time-bounded channels, one can define \emph{online
  space-bounded} channels. An online space-$S$ channel is a
width-$2^S$ branching program that makes a single in-order pass over
the transmitted codeword, and outputs one bit for every bit that is
read. Such channels were first considered by Galil \emph{et
  al.}~\cite{GLYY} and are special case of both time-bounded channels
(since a space-$S$ channel can be implemented by a time-$N\cdot 2^S$
circuit) and of \emph{causal channels}~\cite{DJL08,LJD09}.

Another direction for future work is the construction of fully explicit
codes for space-bounded channels, without hardness assumptions or
Monte Carlo constructions. We considered online space-bounded channels
in the initial version of this paper. The analysis of the codes for
such channels was complex, because the reductions needed to preserve
logarithmic space.
% A referee pointed out that the Monte Carlo code construction in our initial paper in fact tolerates time-bounded channels, with a simpler analysis.
 Our original analysis, however,
shows that one only needs to construct logarithmic-length,
pseudorandom stochastic codes for logarithmic-space channels in order
to get a full construction (as one can use Nisan's generator \cite{nisan} for
$\Gamma$). The lemmas required for obtaining the full construction from
these pieces can be found in version 3 of the arxiv version of this
paper (\url{http://arxiv.org/abs/1004.4017v3}).

% \anote{This is old text.}
% We can define a similar notion of pseudorandom codes for any class of
% channels (for example, space-bounded channels). \anote{Maybe just
%   formlate definition using a class of distinguishers?}
% One open question we
% formulate in Section \ref{sec:openq} is whether one can design explicit list-decodable
% pseudorandom codes for interesting channel classes (say logspace
% channels or time-bounded channels), possibly under appropriate
% computational assumptions.

% However, an analogous
% constructions for \emph{space-bounded} channels might be possible
%  without assumptions, using current techniques. It would imply a fully
%  explicit construction of capacity-achieving codes for space-bounded
%  channels. 

% \begin{notes}
%   To do:
%   \begin{enumerate}
% \item State explicit open question about codes for space-bounded
%   channels where the decoder runs in \emph{fixed} polynomial time but
%   spaces much larger than $S$ (e.g. linear, or ideally $O(S)$).
%   \item State explicit open question about log-space generators that
%     are also codes (can't be that hard...). Extra credit for codes
%     that can be decoded in a fixed polynomial amount of time (not
%     $2^{O(S)}$), since I think they would answer the first open question.
%   \end{enumerate}
% \end{notes}

\section*{Acknowledgments}

We are extremely grateful to an anonymous referee for observations
that simultaneously strengthened our results and simplified the
analysis of our code constructions. Specifically, the referee pointed
out that the separate constructions we had for space- and time-bounded channels
could be unified to obtain a single, unconditional result covering
both types of channels.

\addcontentsline{toc}{section}{References}
\bibliographystyle{abbrv}
\bibliography{avc}

\appendix

\section{Ingredients for Code Construction for Additive Errors}
\label{sec:ingredients}
In this section, we will describe the various ingredients that we will
need in our construction of capacity achieving AVC codes, expanding on the brief mention of these from Section~\ref{sec:ingre-brief}.

\subsection{Constant rate codes for average error}
By plugging in an appropriate explicit construction of list-decodable
codes (with sub-optimal rate) into Theorem~\ref{thm:ld-to-avc}, we can
also get the following explicit constructions of stochastic codes,
albeit not at capacity. We will make use of these codes to encode
blocks of logarithmic length control information in our final
capacity-achieving explicit construction. The total number of bits in
all these control blocks together will only be a small fraction of the
total message length. So the stochastic codes encoding these blocks
can have any constant rate, and this allows us to use any
off-the-shelf explicit constant rate list-decodable code in
Theorem~\ref{thm:ld-to-avc} (in particular, we do {\em not} need a
brute-force search for small list-decodable codes of logarithmic block
length). We get the following claim by choosing $d=1$ and picking $C$
to be a binary linear $(\alpha,c_1(\alpha)/2)$-list decodable code in
Theorem~\ref{thm:ld-to-avc}.
\begin{prop}
\label{prop:const-rate-AVC}
For every $\alpha$, $0 < \alpha < 1/2$, there exists $c_0 =
c_0(\alpha) > 0$ and $c_1 =c_1(\alpha) < \infty$ such that for all
large enough integers $b$, there is an explicit stochastic code
$\SC_{k,\alpha}$ of rate $1/c_0$ with encoding $E : \{0,1\}^b \times
\{0,1\}^b \rightarrow \{0,1\}^{c_0 b}$ that is efficiently
strongly $\alpha$-decodable with probability $1- c_12^{-b}$. 

Moreover, for every
message and every error pattern of more than a fraction $\alpha$ of
errors, the decoder for $\SC_{k,\alpha}$ returns $\bot$ and reports a
decoding failure with probability $1-c_1 2^{-b}$.  

Further, there exists an absolute constant $c_3 = c_3(\alpha)$ such
that on input a uniformly random string $y$ from $\{0,1\}^{c_0 b}$,
the decoder for $\SC_{k,\alpha}$ returns $\bot$ with probability at
least %$1-2^{-c_3 b}$
$1-c_12^{-b}$
 (over the choice of $y$).
\end{prop}

\begin{proof}
  The claim follows by choosing $d=1$ and picking $C$ to be a binary
  linear $(\alpha,c_1(\alpha)/2)$-list decodable code in
  Theorem~\ref{thm:ld-to-avc}. The claim about decoding a uniformly
  random input follows since the number of strings $y$ which differ
  from some valid output of the encoder
  $E$ %codeword of the list-decodable code $C$
  is at most a fraction $\alpha$ of positions is at most $2^{2b}
  2^{H(\alpha) c_0 b}$. By standard entropy arguments, we have  $(1-H(\alpha))c_0b+\log(c_1(\alpha)/2)\geq 3b$ (since the code encodes $3b$ bits, the capacity is $1-H(\alpha)$, and at most $\log(c_1(\alpha)/2)$ additional bits of side information are necessary to disambiguate the true message from the list).  We conclude that the probability that a random string gets accepted by the decoder is at most $2^{-b}\cdot 2^{\log( c_1(\alpha)/2)} \leq c_12^{-b}$.
% \le 2^{(1-c_3) c_0 b}$ for small enough $c_3 =
%  c_3(\alpha)$.
\end{proof}

\iffalse
We already saw one of the ingredients we will need --- a conversion of
list-decodable codes into codes for average error criterion for
$\adv_p$. We will use a constant rate list-decodable code that can
correct any fraction of errors $< 1/2$. This code will be applied to
small (logarithmic length) blocks comprising a small fraction of the
total message length, so their rate can be far from capacity. So we
can use any off-the-shelf explicit list-decodable code. Note that
because of this our final construction will be completely explicit ---
we will not even need an expensive exhaustive search for optimal small
length codes.
\fi

\subsection{Reed-Solomon codes}

If $\F$ is a finite field with at least $n$ elements, and $\spcbound =
(\alpha_1,\alpha_2,\dots,\alpha_n)$ is a sequence of $n$ {\em distinct}
elements from $\F$, the Reed-Solomon encoding, $\RS_{\F,\spcbound,n,k}(m)$, or
just $\RS(m)$ when the other parameters are implied, of a message $m = (m_0,m_1,\dots,m_{k-1}) \in \F^k$ is given by
\begin{equation}
\label{eq:def-RS}
\RS_{\F,\spcbound,n,k}(m) = (f(\alpha_1), f(\alpha_2), \cdots, f(\alpha_n)) \ . 
\end{equation}
where $f(X) = m_0 + m_1X + ... + m_{k-1} X^{k-1}$. The following is a classic result on unique decoding Reed-Solomon codes~\cite{peterson}, stated as a noisy polynomial reconstruction algorithm.

\begin{prop}[Unique decoding of RS codes]
\label{prop:RS}
There is an efficient algorithm with running time polynomial in $n$
and $\log |\F|$ that given $n$ distinct pairs $(\alpha_i,a_i) \in
\F^2$, $1\le i \le n$, and an integer $k < n$, finds the unique
polynomial $f$ of degree at most $k$, if any, that satisfies
$f(\alpha_i) = a_i$ for more than $\frac{n+k}{2}$ values of $i$.  Note
that this condition can also be expressed as $\bigl| \{i : f(\alpha_i)
= a_i \} \bigr| - \bigl| \{ i : f(\alpha_i)\neq a_i) \} \bigr| > k$.
\end{prop}

We also state a list-decoding generalization (the version due to Sudan~\cite{Sudan} suffices for our purposes), which will be used in our result for space-bounded channels.

\begin{prop}[List decoding of RS codes~\cite{Sudan}]
\label{prop:RS-list}
There is an efficient algorithm with running time polynomial in $n$
and $\log |\F|$ that given $n$ distinct pairs $(\alpha_i,a_i) \in
\F^2$, $1\le i \le n$, and integer $k < n$, finds the set $\cal L$ of
all polynomials $f$ of degree at most $k$, if any, that satisfy
$f(\alpha_i) = a_i$ for at least $t$ values of $i$ as long as $t >
\sqrt{2kn}$. Moreover, there are at most $\sqrt{2n/k}$ polynomials in
the set $\cal L$.
\end{prop}

\subsection{Pseudorandom constructs}

\subsubsection{Samplers}

Let $[N] = \{1,2,\dots,N\}$. If $B \subseteq [N] \rightarrow \{0,1\}$
has density $\mu$ (i.e., $\mu N$ elements), then standard tail bounds
imply that for a random subset $\ctrlset \subseteq [N]$ of size $\ell$, the
density of $B \cap \ctrlset$ is within $\pm \theta$ of $\mu$ with
overwhelming probability (at least $1-\exp(-c_\theta \ell)$). But picking
a random subset of size $\ell$ requires $\approx \ell \log (N/\ell)$ random
bits. The following shows that a similar effect can be achieved by a
sampling procedure that uses fewer random bits. The idea is the well
known one of using random walks of length $\ell$ in a low-degree expander
on $N$ vertices. This could lead to repeated samples while we would
like $\ell$ distinct samples. This can be achieved by picking slightly
more than $\ell$ samples and discarding the repeated ones. The result
below appears in this form as Lemma 8.2 in \cite{vadhan04}.

\begin{prop}
\label{prop:sampler}
  For every $N \in {\mathbb N}$, $0 < \theta < \mu < 1$, $\gamma > 0$,
  and integer $\ell \ge \ell_0 = \Omega(\frac{1}{\theta^2} \log
  (1/\gamma))$, there exists an explicit efficiently computable
  function $\Samp : \{0,1\}^\sigma \rightarrow [N]^\ell$ where $\sigma \le
  O(\log N + \ell \log (1/\theta))$ with the following property: 

  For every $B \subseteq [N]$ of size at least $\mu N$, with
  probability at least $1- \gamma$ over the choice of a random $s \in
  \{0,1\}^\sigma$, $|\Samp(s) \cap B| \ge (\mu - \theta) |\Samp(s)|$.

\end{prop}

We will use the above samplers to pick the random positions in which
the blocks holding encoded control information are interspersed with
the data blocks. The sampling guarantee will ensure that a reasonable
fraction of the control blocks have no more than a fraction $p+\eps$
of errors when the total fraction of errors is at most $p$.

\subsubsection{Almost $t$-wise independent permutations}
\begin{defn}
  A distribution ${\cal D}$ on $\spcbound_n$ (the set of permutations of
  $\{1,2,\dots,n\}$) is said to {\em almost $t$-wise independent} if
  for every $1 \le i_1 < i_2 < \cdots < i_t \le n$, the distribution  of $(\pi(i_1),\pi(i_2),\dots,\pi(i_t))$ for $\pi$ chosen according
  to ${\cal D}$ has statistical distance at most $2^{-t}$ for the
  uniform distribution on $t$-tuples of $t$ distinct elements from
  $\{1,2,\dots,n\}$. \closedef
\end{defn}
A uniformly random permutation of $\{1,2,\dots,n\}$ takes $\log n!=
\Theta(n \log n)$ bits to describe. The following result shows that
almost $t$-wise independent permutations can have much shorter
descriptions.
\begin{prop}[\cite{KNR}]
\label{prop:indep-perms}
  For all integers $1 \le t \le n$, there exists $D = O(t \log n)$ and
  an explicit map $\knr : \{0,1\}^\sigma\rightarrow \spcbound_n$, computable in
  time polynomial in $n$, such that the distribution $\knr(s)$ for
  random $s \in \{0,1\}^{\sigma}$ is almost $t$-wise independent.
\end{prop}
%Give the KNR definition and state the result for $t= \Omega(n/log n)$.

\subsubsection{$t$-wise independent bit strings}% sample spaces}
%
%\vnote{Looks like we can even afford perfect $t$-wise independence,
%  since we are okay with seed length $O(t \log n)$, so that is what I
%q  have defined.}
%
We will also need small sample spaces of binary strings in $\{0,1\}^n$
which look uniform for any $t$ positions.
\begin{defn}
  A distribution ${\cal D}$ on $\{0,1\}^n$ is said to {\em $t$-wise independent}
  if for every $1 \le i_1 < i_2 < \cdots < i_t \le n$, the
  distribution of $(x_{i_1},x_{i_2},\dots,x_{i_t})$ for $x =
  (x_1,x_2,\dots,x_n)$ chosen according to ${\cal D}$ equals the uniform distribution on $\{0,1\}^t$. \closedef
\end{defn}
Using evaluations of degree $t$ polynomials over a field of
characteristic $2$, the following well known fact can be shown. We
remark that the optimal seed length is about $\frac{t}{2} \log n$ and
was achieved in \cite{ABI86}, but we can work with the weaker $O(t \log
n)$ seed length.
\begin{prop}
\label{prop:twise}
Let $n$ be a positive integer, and let $t \le n$. There exists $\sigma \le O( t \log n)$ and an explicit map $\twise_t: \{0,1\}^{\sigma} \rightarrow \{0,1\}^n$, computable in time polynomial in $n$, such that the distribution $\twise_t(s)$ for random $s \in \{0,1\}^{\sigma}$ is $t$-wise independent.
\end{prop}

\iffalse
\begin{defn}
  A distribution ${\cal D}$ on $\{0,1\}^n$ is said to {\em almost
    $t$-wise independent} if for every $1 \le i_1 < i_2 < \cdots < i_t
  \le n$, the distribution of $(x_{i_1},x_{i_2},\dots,x_{i_t})$ for $x
  = (x_1,x_2,\dots,x_n)$ chosen according to ${\cal D}$ has
  statistical distance at most $2^{-10t}$ from the uniform
  distribution on $\{0,1\}^t$.
\end{defn}

\begin{prop}[\cite{AGHP}]
Let $n$ be a positive integer, and let $t \le n$. There exists $\sigma = O( t + \log n)$ and an explicit map $\aghp : \{0,1\}^{\sigma} \rightarrow \{0,1\}^n$, computable in time polynomial in $n$, such that the distribution $\aghp(s)$ for random $s \in \{0,1\}^{\sigma}$ is almost $t$-wise independent.
\end{prop}
\fi

\subsection{Capacity achieving codes for $t$-wise independent errors}

Forney~\cite{forney} constructed binary linear concatenated codes that
achieve the capacity of the binary symmetric channel
$\BSC_p$. Smith~\cite{smith07} showed that these codes also correct
patterns of at most a fraction $p$ of errors w.h.p. when the error
locations are distributed in a $t$-wise independent manner for large
enough $t$. The precise result is the following.  

\begin{prop}
\label{prop:concat-codes}
For every $p$, $0 < p < 1/2$ and every $\eps > 0$, there is an
explicit family of binary linear codes of rate $R \ge 1- H(p) -\eps$
such that a code $\REC : \{0,1\}^{Rn} \rightarrow \{0,1\}^n$ of block
length $n$ in the family provides the following guarantee. There is a
polynomial time decoding algorithm $\Dec$ such that for every message
$m \in \{0,1\}^{Rn}$, every error vector $e \in \{0,1\}^n$ of Hamming
weight at most $pn$, and every almost $t$-wise independent
distribution ${\cal D}$ of permutations of $\{1,2,\dots,n\}$, we have
\[ \Dec(\REC(m)+\pi(e)) = m \] 
with probability at least $1-2^{-\Omega(\eps^2 t)}$ over the choice of
a permutation $\pi \in_R {\cal D}$, as long as $\omega(\log n) < t<
\eps n /10$. (Here $\pi(e)$ denotes the permuted vector: $\pi(e)_i =
e_{\pi(i)}$.)
\end{prop}

We will use the above codes (which we denote $\REC$, for
``random-error code'') to encode the actual data in our stochastic
code construction.

\section{Capacity-achieving codes for average error}
\label{app:codes-avg-error}
The average error criterion is an extensively studied topic in the
literature on arbitrarily varying channels; see the survey
\cite{LN-avc-survey} and the many references therein. Here we assume
the message is unknown to the channel and the decoding error
probability is taken over a uniformly random choice of the message and
the noise of the channel. The following defines this notion for the
special case of the additive errors. The idea is that we want every
error vector to be bad for only a small fraction of messages.

\begin{defn}[Codes for average error]
\label{def:avc-code}
  A code $C$ with encoding function ${\cal E} : {\cal M} \rightarrow
  \Sigma^n$ is said to be (efficiently) {\em $p$-decodable
    with average error $\delta$} if there is a (polynomial time
  computable) decoding function $D: \Sigma^n \rightarrow {\cal M} \cup
  \{ \bot \}$ such that for {\em every} error vector $e \in \Sigma^n$,
  the following holds for at least a fraction $(1-\delta)$ of messages
  $m \in {\cal M}$: $D({\cal E}(m)+ e)= m$. \closedef 
%  In this case, we also say that $C$ allows reliable communication on
% $\obl_p$ with average error $\delta$.
\end{defn}

\subsection{Codes for average error from stochastic codes for additive errors}

A slightly more general notion of  stochastic codes
(Definition~\ref{def:stochastic}) implies codes for average error.

\begin{defn}[Strongly decodable stochastic codes]
\label{def:ssc}
  We say a stochastic code is strongly
  $p$-decodable with probability $1-\delta$  if the decoding
  function correctly computes \emph{both} the message $m$ and
  randomness $\omega$ used at the encoder, with probability at least
  $1-\delta$. \closedef
\end{defn}

Using a {\em strongly decodable} stochastic code we can get a code for
average error by simply using the last few bits of the message as the
randomness of the stochastic encoder. If the number of random bits
used by the stochastic code is small compared to the message length,
the rates of the codes in the two models are almost the same.
\begin{observation}
\label{obs:ssc-avc}
  A stochastic code $\mathsf{SSC}$ that is strongly $p$-decodable with
  probability $1-\delta$ gives a code $\mathsf{AVC}$ of the same block
  length that is $p$-decodable with average error $\delta$. If the
  ratio of number of random bits to message bits in $\mathsf{SSC}$ is
  $\lambda$, the rate of $\mathsf{AVC}$ is $(1+\lambda)$ times the
  rate of $\mathsf{SSC}$.
\end{observation}

\subsection{Explicit capacity-achieving codes for average error}
\label{app:explicit-avg-error}
We would now like to apply Observation~\ref{obs:ssc-avc} to the
stochastic codes constructed in Section~\ref{sec:efficient-oc-codes}
and also construct explicit codes achieving capacity for the average
error criterion. For this, we need to ensure that the decoder for the
stochastic code can also recover all the random bits used at the
encoding. We already showed (Lemma~\ref{lem:controlinfo}) that the
random string $\omega$ comprising the control information is in fact
correctly recovered w.h.p. However, there is no hope to recover all
the random strings $r_1,r_2,\dots,r_\ell$ used by the various $\SC$
encodings. This is because some of these control blocks could incur
much more than a fraction $p + \eps$ of errors (or in fact be totally
corrupted).

Our idea is to use the {\em same} random string $r$ for each of the
$\ell$ encodings $\SC(A_i,r)$ in Step~\ref{state:sc-encoding}. Since
each run of {\sc SC-Decode} is correct with probability at least
$1-c_1/N^2$, by a union bound over all $n$ blocks, we can claim that
all the following events occur with probability at least $1- c_1/N$
(over the choice of $r$):
\begin{quote}
  Among the control blocks, all of the at least $\eps\ell/2$ control
  blocks with at most a fraction $p+\eps$ of errors are decoded
  correctly, along with the random string $r$, by {\sc
    SC-Decode}. Further, {\sc SC-Decode} outputs $\bot$ on all the
  other control blocks.  Thus the correct random string $r$ gets at
  least $\eps \ell/2$ ``votes.''
\end{quote}
By Lemma~\ref{lem:payloadblocks}, with probability at least
$1-\exp(-\Omega(\eps^2 N/ \log^2 N)))$ (over the choice of $\omega$),
the number of payload blocks that get accepted as control blocks is at
most $\eps \ell/24$. (Note that this lemma only used the $t'$-wise independence of the offset string $\Delta$.)

The above facts imply that the control information $\omega$ is
recovered correctly with probability at least $1-O(1/N)$ over the
choice of $(\omega,r)$ (this is the analog of
Lemma~\ref{lem:controlinfo}).  Also $r$ is the unique string which
will get at least $\eps \ell/2$ votes from the various runs of {\sc
  SC-Decode}. Therefore it can be correctly identified (with
probability at least $1-O(1/N)$ over the choice of $(\omega,r)$) after
running $\text{\sc SC-Decode}$ on all the $n$ blocks. We can thus conclude
the following result on capacity-achieving codes for average error
(Definition~\ref{def:avc-code}).
\begin{lemma}[Polynomially small average error]
\label{lem:efficient-avc-code}
For every $p\in(0,1/2)$, and every $\eps>0$, there is an explicit
family of binary codes of rate at least $1-H(p)-\eps$ that are
efficiently $p$-decodable with average error $O(1/N)$ where $N$ is the
block length of the code.
\end{lemma}

One can reduce the error probability in this theorem by using redundant, but $t$-wise independent, values $r_i$ for the control block encodings. Specifically, let $(r_1,...,r_\ell)$ be a random codeword from a Reed-Solomon code of dimension $\eps \ell /8$ (the simpler construction above corresponds to a majority code). Then the $r_i$ values are, in particular, $\eps \ell /8$-wise independent. One can modify the proof of Lemma~\ref{lem:controlblocks} (which states that sufficiently many control blocks are recovered) to rely on only this limited independence. Under the same conditions that the control information is correctly recovered, there is enough information to recover the entire vector $r_1,...,r_\ell$. We can thus prove the following:
\begin{theorem}[Exponentially small average error]
\label{thm:efficient-avc-code}
For every $p\in(0,1/2)$, and every $\eps>0$, there is an explicit
family of binary codes of rate at least $1-H(p)-\eps$ that are
efficiently $p$-decodable with average error
$\exp(-\Omega_\eps(N/\log^2 N))$ where $N$ is the block length of the
code.
\end{theorem}

\section{Impossibility Results for Bit-Fixing Channels when $p> \frac 1 4$}\label{sec:imposs}

We show that even very simple channels prevent reliable communication if they can introduce a fraction errors strictly greater than $1/4$. In particular, this result (a) separates the additive (i.e., oblivious) error model from bounded-space channels when $p>1/4$, and (b) shows that some relaxation of correctness is necessary to handle space- and time-bounded channels when $p>1/4$.

\begin{theorem}[Impossibility for $p>\frac 1 4$, detailed version]\label{thm:imposs} For every pair of randomized encoding/decoding algorithms $\Enc,\Dec$ that make $n$ uses of the channel and use a message space whose size tends to infinity with $n$, if a uniformly random message is sent over the channel, then
  \begin{enumerate}\item 
    there is a distribution over memoryless channels that alters at most
    $n/4$ bits \emph{in expectation} and causes a decoding error with
    probability at least $\frac 1 2 -o(1)$.%\anote{Can we get this with just one advice string?}
  \item for every $0<\nu<\frac 1 4$, there is an online
    space-$\lceil\log(n)\rceil$ channel $\channel_2$ that alters at
    most $ n(\frac 1 4 + \nu)$ bits (with probability 1) and
    causes a decoding error with probability $\Omega(\nu)$.
  \end{enumerate}
\end{theorem}

Our proof adapts the impossibility results of Ahlswede~\cite{ahlswede}
on arbitrarily-varying channels. We present a self-contained proof for
completeness. Readers familiar with the AVCs literature will recognize
the idea of \emph{symmetrizability}  from \cite{ahlswede}. 

\mypar{The Swapping Channel}
We begin by considering a simple \emph{swapping} channel, whose
behavior is specified by a \emph{state} vector $s =(s_1,...,s_n) \in
\bit{n}$.  On input a transmitted word $c=(c_1,...,c_n) \in\bit{n}$,
the channel $\channel_s$ outputs $c_i$ in all positions where
$c_i=s_i$, and a random bit in all positions where $c_i\neq s_i$. The
bits selected randomly by the channel at different positions are
independent.

There are several equivalent characterizations that help to understand
the channel's behavior. First, we may view the channel as outputting
either $c_i$ or $s_i$, independently for each
position. 
%On the other hand, if the distance between $s$ and $c$ is known ahead of time, we can think of the channel as deciding on a set of positions ahead of time which it will set to the value $s_i$, and leaving the remaining positions untouched.

$$\channel_s(c)_i =
\begin{cases}
  c_i & \text{if }c_i=s_i \\
  U\from \bit{} & \text{if }c_i \neq s_i
\end{cases}
= 
\begin{cases}
  c_i & \text{with prob. }1/2\\ 
  s_i & \text{with prob. }1/2\\ 
\end{cases}
$$

This view of the channel makes it obvious that the output distribution is symmetric with respect to the inversion of $c$ and $s$. That is, 

\begin{equation}
  \label{eq:swap}
  \channel_s(c) \text{ and } \channel_c(s) \text{are identically distributed}
\end{equation}

The key idea behind our lower bounds is that if $s$ is itself a valid codeword, then the decoder cannot tell whether $c$ was sent with state $s$, or $s$ was sent with state $c$. If $c$ and $s$ code different messages, then the decoder will make a mistake with probability at least 1/2.

Note that the expected number of errors introduced by the channel is half of the Hamming distance $\hamdist(c,s)$; specifically, the number of errors is distributed as Binomial$(\hamdist(c,s),\frac 1 2)$. As long as $\hamdist(c,s)$ is close to $n/2$, then the number of errors will be less than $n(\frac 1 4 + \nu)$ with high probability.

\mypar{Hard Channel Distributions} 
Given an stochastic encoder $\Enc(\cdot;\cdot)$, consider the following distribution on swapping channels: pick a random codeword in the image of $\Enc$ and use it as the state.
$$\channel^{main}(c) :
\begin{cases}
 \text{Select }m',r'\text{ uniformly at random}\\
 \text{Compute }s\gets \Enc(m',r') \\
\text{Output }\channel_s(c)
\end{cases}$$

\begin{lemma}Under the conditions of Theorem~\ref{thm:imposs}, for channel $\channel^{main}$:\\
(a)  The probability of a decoding error on a random message is $\frac 1 2 -o(1)$.
\\ 
(b) The expected number of bits altered by $\channel^{(main)}$ is at most $n/4$.
\end{lemma}
\begin{proof}
  (a) We are interested in bounding the probability of a decoding error:
  \begin{multline*}
    \Pr(\text{correct decoding})= \Pr_{{m,r \atop \text{channel
          coins}}}\Big(\Dec(\channel^{main}(\Enc(m,r)))=m\Big) 
\\ =
    \Pr_{{m,r, m',r' \atop \text{swapping
          coins}}}\Big(\Dec(\channel_{\Enc(m',r')}(\Enc(m,r)))=m\Big)\,.
  \end{multline*}

Because of the symmetry of the swapping channel, the right hand side is equal to the probability that the decoder outputs $m'$, rather than $m$. This is a decoding error as long as $m'$ differs from $m$. We assumed that the size of the message space grows with $n$, so the probability that $m=m'$ goes to 0 with $n$. We use ``right'' and ``wrong'' and shorthand for the events that decoding is correct and incorrect, respectively.

$$\Pr(\text{right}) = \Pr_{m,m'}(\text{decoder outputs }m') \leq \Pr(\text{wrong} \vee m=m') \leq \Pr(\text{wrong}) +o(1)\,.
$$
Thus, the probability of correct decoding is at most $\frac 1 2 -o(1)$. This proves part (a) of the Lemma.

It remains to show that the expected number of bit corruptions is at most $n/4$. This follows directly from the following fact, which is essentially the Plotkin bound from coding theory: 

\begin{claim}
  If $(m,r)$ is independent of and identically distributed to
  $(m',r')$, then the expectation of the distance
  $\hamdist(\Enc(m,r),\Enc(m',r'))$  is at most $n/2$.
\end{claim}

\begin{proof}
  By linearity of expectation, the expected Hamming distance is the sum, over positions $i$, of the probability that $\Enc(m,r)$ and $\Enc(m',r')$ disagree in the $i$th positions. The probability that two i.i.d. bits disagree is at most $\frac 12$, so the expected distance is at most $\frac n 2$.
\end{proof}

Part (b) of the lemma follows since the expected number of errors introduced by the swapping channel is half of the Hamming distance between the transmitted word and the state vector. \end{proof}

\mypar{Bounding the Number of Errors} 
To prove part (2) of Theorem~\ref{thm:imposs}, we will find a
(nonuniform) channel with a hard bound on the number of bits it
alters. In logarithmic space, it is easy for the channel to count the
number of bits it has flipped so far and stop altering bits when a
threshold has been exceeded. The difficult part is to show that such a
channel will still cause a significant probability of decryption
error.

As before, the channel will select $m',r'$ at random and run the
swapping channel $\channel_s$ with state $s=\Enc(m',r')$. In addition,
however, it will stop altering bits once the threshold of $n(\frac 1 4 +\nu)$ bits have
been exceeded.

Consider now the transmission of a random codeword $c=\Enc(m,r)$. Let
$G$ be the event that $\hamdist(c,s) \leq n(\frac 12+\nu )$. By a
Markov bound, the probability of $\overline{G}$ is at most
$\frac{1/2}{1/2+\nu}$, and so the probability of $G$ is $1-\Pr(\bar G)
\geq \frac{2\nu}{1+2\nu}\geq \nu$. Conditioned on $G$, the number of
bits altered by $W_s$ on input $c$ is dominated by Binomial$(n(\frac
12 + \nu),\frac 12)$. The probability that the number of bits altered
exceeds $n(\frac 1 4 + \nu)$ is therefore at most $\exp(-\Omega(\nu^2 n))$.

On the other hand, conditioned on $G$ there is a significant
probability of a decoding error. To see why this is the case, first note that
conditioned on $G$ the error-bounded channel will simulate
$\channel_s(c)$ nearly perfectly. Moreover, the event $G$ is symmetric
in $c$ and $s$, and so conditioning on $G$ does not help to
distinguish $\channel_c(s)$ from $\channel_s(c)$. By the same
reasoning as in the previous proof,
$$\Pr(\text{incorrect decoding}|G) \geq \frac 1 2 -o(1)\,.$$
Since $G$ has probability at least $\nu$, the channel causes a
decoding error with probability at least $\frac \nu 2 -o(1)$, in expectation over the choice of $s$. Hence, there exists a specific string $s^*$ for which the channel causes a decoding error with probability $\frac \nu 2 -o(1)$. This completes the proof of Theorem~\ref{thm:imposs}.

\end{document}